%% file: main.tex
\documentclass[nonhyperref]{article}

\usepackage{microtype}
\usepackage{graphicx}
\usepackage{subfigure}
\usepackage{booktabs} %
\usepackage{csquotes}

\usepackage{hyperref}

\usepackage[accepted]{icml2024}

\usepackage{amsmath}
\usepackage{amssymb}
\usepackage{mathtools}
\usepackage{amsthm}
\usepackage{algorithm}
\usepackage[noend]{algpseudocode}

\usepackage{dsfont}

\usepackage[capitalize,noabbrev]{cleveref}

\usepackage[noend]{algpseudocode}
\usepackage{pgfplots}
\theoremstyle{plain}
\newtheorem{theorem}{Theorem}[section]

\newtheorem{lemma}[theorem]{Lemma}

\newtheorem{claim}[theorem]{Claim}
\theoremstyle{definition}
\newtheorem{definition}[theorem]{Definition}

\theoremstyle{remark}

\input{macros.tex}

\newenvironment{proofof}[1]{\begin{trivlist} \item {\bf Proof
#1:~~}}
  {\qed\end{trivlist}}

\usepackage[textsize=tiny]{todonotes}

\icmltitlerunning{A Dynamic Algorithms for Weighted Submodular Cover Problem}

\begin{document}

\icmltitle{A Dynamic Algorithm for Weighted Submodular Cover Problem}

\icmlsetsymbol{equal}{*}

\begin{icmlauthorlist}
\icmlauthor{Kiarash Banihashem}{equal,yyy}
\icmlauthor{Samira Goudarzi}{equal,yyy}
\icmlauthor{MohammadTaghi Hajiaghayi}{equal,yyy}
\icmlauthor{Peyman Jabbarzade}{equal,yyy}
\icmlauthor{Morteza Monemizadeh}{equal,comp}

\end{icmlauthorlist}

\icmlaffiliation{yyy}{Department of Computer Science, University of Maryland, MD, USA}
\icmlaffiliation{comp}{Department of Mathematics and Computer Science, TU Eindhoven, the Netherlands}
\icmlcorrespondingauthor{Samira Goudarzi}{samirag@umd.edu}
\icmlcorrespondingauthor{Kiarash Banihashem}{kiarash@umd.edu}

\icmlkeywords{Dynamic Algorithm, Submodular Optimization, Machine Learning, ICML}

\vskip 0.3in

\printAffiliationsAndNotice{\icmlEqualContribution} %

\begin{abstract}
  We initiate the study of the submodular cover problem in dynamic setting where the elements of the ground set are inserted and deleted.
  In the classical submodular cover problem, we are given a monotone submodular function $f : 2^{V} \to \mathbb{R}^{\ge 0}$ and the goal is to obtain a set $S \subseteq V$ that minimizes the cost subject to the constraint $f(S) = f(V)$. This is a classical problem in computer science and generalizes the Set Cover problem, 2-Set Cover, and dominating set problem among others.
  We consider this problem in a dynamic setting where there are  updates to our set $V$, in the form of insertions and deletions of elements from a ground set $\mathcal{V}$, and the goal is to maintain an approximately optimal solution with low query complexity per update. For this problem, we propose a randomized algorithm that, in expectation, obtains a $(1-O(\epsilon), O(\epsilon^{-1}))$-bicriteria approximation using polylogarithmic query complexity per update.

\end{abstract}

\section{Introduction}
Submodular optimization is a classical problem in computer science and machine learning
with applications spanning various domains such as data summarization, active learning, network inference, video analysis, and facility location (see \cite{DBLP:conf/bmvc/Krause13} for a survey).

The submodular cover problem, initially introduced by \cite{wolsey1982analysis},
is a well-studied classical variant of the problem where the objective is to minimize the sum of the weight of selected elements chosen from a set subject to a submodular function constraint. 
Specifically, given a set of elements $V$, a monotone submodular function $f: 2^{V} \to \mathbb{R}^{\ge 0}$,
and a weight function $w: V \to \mathbb{R}^{\ge 0}$, we seek to
pick a set $S$ minimizing $\sum_{v \in S} w(v)$ that satisfies $f(S) = f(V)$.

This problem generalizes various noteworthy problems such as the set cover problem, 2-set cover, dominating set, and others. It can also be seen as a dual of the submodular maximization problem, in which the goal is to maximize $f(S)$ subject to the constraint $|S| \le k$ for some parameter $k$.

While the submodular cover problem has been extensively studied (see \cite{bar2001generalized} for a survey), the majority of the algorithms in the literature predominantly depend on having access to the entire ground set throughout their execution, which is not a valid assumption in numerous real-world applications dealing with ever-changing data and makes them impractical.

Given the mentioned limitation, there has recently been a surge of interest in reexamining classical problems under a variety of massive data models such as streaming, distributed, dynamic, and online settings. 
For submodular maximization, the problem has been considered 
in the streaming~\cite{badanidiyuru2014streaming,chakrabarti2015submodular, mirzasoleiman2018streaming, DBLP:conf/icml/0001MZLK19},
distributed~\cite{DBLP:conf/stoc/MirrokniZ15, liu2018submodular},
and dynamic~\cite{DBLP:conf/nips/Monemizadeh20,DBLP:conf/nips/LattanziMNTZ20,chen2022complexity,DBLP:conf/icml/DuettingFLNZ23, banihashem2023dynamic, banihashem2023nonmonotone,banihashem2023dynamicmat}
settings.
Similarly for submodular cover, recent works have studied
the problem in the distributed, streaming, and scalable settings~\cite{mirzasoleiman2015distributed, norouzi2016efficient, chen2023bicriteria, crawford2023scalable}.

Motivated by these advances, we consider the submodular cover in a dynamic setting where the elements of the ground set are inserted and deleted, and the goal is to always maintain an approximately optimal solution.
While this can easily be done by re-running an offline algorithm after each update, the goal is to do this with small update time per query. 
We formally define the dynamic submodular cover problem as follows.

\begin{definition}[Dynamic Submodular Cover problem]
    \label{def:dyn:sub:cover}
    We assume that $f: 2^{\ground} \rightarrow  \mathbb{R}^+$ is a monotone, non-negative submodular function on the ground set $\ground = \{v_1, v_2, \dots, v_n\}$, and each element $v$ in the $\ground$ has a weight denoted by $\weight(v)$. For any subset $S \subseteq \ground$, $\cost(S)$ is defined to be $\sum_{v \in S}{\weight(v)}$. 
    At each time $t$, the objective of the problem is to choose a subset $S \subseteq {V_t}$ of minimum cost whose submodular value is equal to $f({V_t})$, i.e.,
 \[
        \solution_{\opt_t} = \arg\min_{S \subseteq V_t } \{ \cost(S): f(S) = f(V_t) \} \enspace, 
 \]
where $V_t$ denotes the set of the currently present elements after the first $t$ updates. $\text{OPT}_{\text{cost}_t}$ is defined to be $ \cost(\solution_{\opt_t})$, and $f(\solution_{\opt_t}) = f({V_t})$.
\end{definition}
Note that, throughout the paper, as we consider a fixed point of time, we drop the subscript $t$ for simplicity.

We note that while \cite{gupta2020fully} also considered the submodular cover problem
in a dynamic setting, their model is different as it assumes that the submodular function $f$ is changing dynamically, whereas we assume that the ground set undergoes updates.
To illustrate the difference, consider the special case of set cover where the elements of $V$ correspond to sets in a set system and 
we define $f(\mathcal{S}) := |\cup_{S \in \mathcal{S}} S|$ for any $\mathcal{S} \subseteq V$.
In this case, the model in \cite{gupta2020fully} assumes that the elements of the set system are inserted and deleted, while our model assumes that sets of the set system are inserted and deleted.
Our model is consistent with the models considered for the streaming version of the problem~\cite{norouzi2016efficient} where the elements are inserted one by one (but the elements are never deleted in their setting) and
dynamic setting considered for the submodular maximization~ problem\cite{DBLP:conf/nips/Monemizadeh20,DBLP:conf/nips/LattanziMNTZ20,chen2022complexity,DBLP:conf/icml/DuettingFLNZ23, banihashem2023dynamic, banihashem2023nonmonotone,banihashem2023dynamicmat}.

In this paper, we assume that the updates are specified by an \emph{oblivious adversary}, that is
an adversary who knows are algorithm but does not
have access to the random bits we use. This is equivalent to assuming that all of the updates are specified before the algorithm is run and as such are not adapted to the algorithm's output.

\subsection{Our Contribution} 
In this paper, we design an algorithm for the dynamic submodular cover problem that maintains an approximately
optimal solution using polylogarithmic update time.
As is standard for the submodular cover problem~\cite{norouzi2016efficient,chen2023bicriteria}, our
approximation guarantees are bicriteria given the two objectives of the problem.
A set $S$ is called a $(1-\epsilon, c)$-bicriteria approximate solution if it satisfies
\begin{align*}
    f(S) \ge (1-\epsilon) f(V),
    \quad
    \text{and}
    \quad
    \cost{}(S) \le c \, \cost{}(\solution_{\opt}),
\end{align*}
where $\solution_{\opt}$ denotes the optimal solution.
We say a (random) set $S$ is \emph{expected} $(1-\epsilon, c)$-bicriteria, if
the first guarantee holds in expectation, i.e.,
\begin{math}
    \Ex{f(S)} \ge (1-\epsilon) f(V).
\end{math}
Our main result is stated in the following theorem.
\begin{theorem}
\label{contributiontheorem}
    Define the \emph{weight ratio} of $\ground$ as 
    \begin{math}
        \rho := \frac{\max_{v \in \ground} w(v)}{\min_{v \in \ground} w(v)},
    \end{math}
    and set  $n := |\ground|$.
    For any $\epsilon > 0$, there is an algorithm for the dynamic submodular cover problem that maintains an
    expected $(1-O(\epsilon), O(\epsilon^{-1}))$-bicriteria approximate solution with
    expected amortized $\poly{}(\log(n), \log(\rho), \eps^{-1})$ update time query complexity.
\end{theorem}

In terms of techniques,
we build on and generalize the recent advances for dynamic submodular maximization~\cite{DBLP:conf/nips/Monemizadeh20,DBLP:conf/nips/LattanziMNTZ20,banihashem2023dynamic},
in particular the
multi-level construction proposed by \cite{banihashem2023dynamic}, but require important changes given the \enquote{two-dimensional} nature of the problem involving both submodular value $f(.)$ and the weights $w(.)$.
Indeed, the underlying \enquote{static} algorithm in our approach can be seen as a generalization of the algorithm in \cite{norouzi2016efficient} that can support arbitrary weights (as opposed to the uniform weight setting of \cite{norouzi2016efficient}). Our bucketing structure is two-dimensional in order to handle the effect of deletions, unlike \cite{banihashem2023dynamic}.  
We handle our parallel runs, and solution retrieval differently.
Additionally, to simplify the analysis of the approximation guarantee, we check the marginal density of each element in the solution at the time of forming our solution sets, as opposed to \cite{banihashem2023dynamic} who add their elements in bulk.\footnote{We note that the same idea is used in the corrected version of \cite{DBLP:conf/nips/LattanziMNTZ20}.} As such, we do not need the complicated potential function analysis used in \cite{banihashem2023dynamic}, which is crucial for simplifying the analysis given our more involved two-dimensional setting.

\section{Related Work} 

\paragraph{Submodular cover}
The offline version of submodular cover has been extensively studied and it is well-known that the greedy algorithm by \cite{wolsey1982analysis} obtains a logarithmic approximation ratio for the problem (see \cite{bar2001generalized} for a survey of developments and applications).
Recently, given the emergence of Big Data algorithms, there has been an interest in considering the problem in the 
streaming setting.
In particular, 
\cite{norouzi2016efficient}
obtain a $(1-\epsilon, O(\epsilon^{-1}))$-bicriteria approximation algorithm for the problem in the special case where the weights are uniform, i.e., each element has weight $1$. 
Here, the elements of the ground set arrive in a stream and the goal is to build an approximately optimal solution with low-memory.
As mentioned earlier, our underlying static algorithm can be seen as a generalization of this approach for handling arbitrary weights.
Subsequent works have considered the non-monotone objective functions and designed scalable algorithms for the problem~\cite{crawford2023scalable,chen2023bicriteria}

For the dynamic setting, \cite{gupta2020fully} consider a different variant of the problem in which the submodular function $f$ changes overtime and obtain a fully dynamic algorithm with bounded recourse. In contrast, our approach assumes that the underlying function is fixed and the ground set changes.
This is aligned with the models considered in the streaming setting~\cite{norouzi2016efficient,crawford2023scalable} as well the models considered for streaming and dynamic settings for the submodular maximization problem~\cite{DBLP:conf/nips/Monemizadeh20, DBLP:conf/nips/LattanziMNTZ20}.

\paragraph{Submodular maximization}
A closely related problem to submodular cover is submodular maximization.
In the classical version of this problem, we are given a ground set of elements $V$, a submodular function $f: 2^{V} \to \mathbb{R}^{\ge 0}$ and a parameter $k$, and the goal is to maximize the function $f$ over all sets $S$ of size at most $k$.
For the offline version of the problem, it is well-known that a standard greedy algorithm that iteratively chooses a remaining element with maximum marginal gain obtains an approximation ratio of $1-1/e$~\cite{nemhauser1978analysis}.
The approximation ratio cannot be improved efficiently under complexity assumptions as shown by \cite{feige1998threshold} via a reduction from set cover.

In the streaming setting,
submodular maximization was first studied by \cite{badanidiyuru2014streaming} who obtained a $1/2 - \epsilon$-approximation algorithm. The $1/2$ bound was later shown to by \cite{DBLP:conf/icml/Norouzi-FardTMZ18}.
The study of the dynamic version of the problem was first initiated independently by \cite{DBLP:conf/nips/Monemizadeh20} and \cite{DBLP:conf/nips/LattanziMNTZ20} who obtain $(1/2 - \epsilon)$-approximation algorithms 
with $O(k^2 \log^2 (n) \epsilon^{-3})$ and $O(\log^8(n) \epsilon^{-6})$ update times respectively.
\footnote{
The original version of \cite{DBLP:conf/nips/LattanziMNTZ20} had correctness issues in the proof pointed by out \cite{banihashem2023dynamic} who provided an alternative algorithm with polylogarithmic update time.
The issues were subsequently fixed by \cite{DBLP:conf/nips/LattanziMNTZ20}.
}
The $1/2$ approximation is essentially tight 
as shown by \cite{chen2022complexity}
using a lower bound construction based on the streaming version of the problem.
Recent works have generalized the dynamic results for non-monotone objectives~\cite{banihashem2023nonmonotone}, as well as matroid constraints~\cite{DBLP:conf/icml/DuettingFLNZ23, banihashem2023dynamicmat}.

\paragraph{Deletion robust algorithms}
A closely related but distinct area to dynamic submodular optimization is the
robust submodular optimization~\cite{DBLP:conf/icml/MirzasoleimanK017, DBLP:conf/icml/0001ZK18,DBLP:conf/icml/DuettingFLNZ22} in which the goal is to obtain a set that is robust to deletions performed by the adversary. The number of deletions is known upfront, is bounded. In contrast, the dynamic model assumes that insertions and deletions are performed arbitrarily and the goal is to always maintain a good solution.

\section{Preliminaries}

\textbf{Notation: } For a natural number, the set $\{1,2,\dots,x\}$ is denoted as $[x]$. Bold letters represent random variables, while their non-bold counterparts denote specific values. For instance, a random variable is denoted as $\textbf{X}$ and its value as $X$. Probability and expectation of a random variable $\textbf{X}$ are represented by $\Pr{\textbf{X}}$ and $\Ex{\textbf{X}}$ respectively. The notation $\Pr{A|B}$ denotes the conditional probability of event $A$ given event $B$. For an event $A$ with nonzero probability and a discrete random variable $\textbf{X}$, the conditional expectation of $\textbf{X}$ given $A$ is denoted as $\Ex{\textbf{X}|A} = \sum_{x} x\cdot \Pr{\textbf{X} = x | A}$. Likewise, for discrete random variables $\textbf{X}$ and $\textbf{Y}$, the conditional expectation of $\textbf{X}$ given $\textbf{Y}$ is denoted as $\Ex{\textbf{X} | \textbf{Y} = y}$.
The \emph{indicator function} of an event $E$ is denoted by $\ind{E}$, where $\ind{E}$ is assigned one if $E$ occurs and zero otherwise.

\textbf{Submodular functions: }
Consider a \emph{non-negative} utility function $f: 2^{\ground} \rightarrow  \mathbb{R}^+$ defined on the given ground set $\ground = \{v_1, v_2, \dots, v_n\}$. For any element $v \in \ground$ and a set $A \subseteq \ground$, $\Delta(v|A)$ is called marginal gain of element $v$ with respect to $S$ and it is defined as $\Delta(v|A) := f(A \cup \{v\}) - f(A)$. Similarly, for any sets $A, B \subseteq \ground$, $\Delta(B|A)$ is defined as $f(A \cup B) - f(A)$. Function $f$ is called \emph{submodular} when for any $A, B \subseteq \ground$, we have $f(A)+f(B) \geq f(A\cup B) + f(A\cap B)$ or equivalently when for any $A \subseteq B \subseteq \ground$ and element $v$, we have $\Delta(v|A) \geq \Delta(v|B)$. This function $f$ is called \emph{monotone} when for any $A \subseteq B \subseteq \ground$, we have $f(A) \leq f(B)$, and it is called \emph{normalized} when $f(\emptyset) = 0$. 

\textbf{Density: }
Given a submodular function $f: 2^{\ground} \rightarrow  \mathbb{R}^+$ defined on the given ground set $\ground = \{v_1, v_2, \dots, v_n\}$, and a weight function $w: {\ground} \rightarrow [1, \rho]$, we define density of an element $v \in \ground$ as $d(v) := \frac{f(v)}{w(v)}$. Similarly, for any set $A \subseteq \ground$ we define marginal density of element $v$ with respect to $A$ as $d(v|A) := \frac{\Delta(v|A)}{w(v)}$.

\textbf{Update time: }
We assume access to the monotone submodular function $f: 2^{\ground} \rightarrow \mathbb{R}^+$ through an \emph{oracle}.
This oracle supports \emph{set queries}, allowing one to inquire about the value $f(A)$ for any subset $A \subseteq \ground$. 
In this paper, we measure running time based on the total number of oracle calls, a common practice in submodular optimization, as the processing time of oracle calls typically dominates the running time of other parts of the algorithm \cite{DBLP:conf/icml/DuettingFLNZ23, banihashem2023dynamicmat}.
We refer to the amortized number of query calls as query complexity and update time.

\section{Dynamic Algorithm}
\label{main:algorithm}

\subsection{Setting}
In the dynamic version of the problem, a sequence of updates, comprising insertions and deletions of elements from the ground set $\ground$, alters the set of the current elements denoted by $V$. Each element may undergo multiple insertions and deletions. At each time frame, the set $V \subseteq \ground$ encompasses all inserted elements that have not been deleted since their last insertion. The algorithm's objective is to maintain a solution after each update, with its performance being assessed by its query complexity for each update.
It is assumed that $\epsilon$ and $\epsDel$ are small enough parameters satisfying $\epsilon < 1/10$ and $\epsDel < \epsilon / 16$.
And lastly, it is also assumed that the weight of each element in $\ground$ is between $1$ and the parameter $\rho$.

\subsection{Overview of the Algorithms}

Our algorithm operates through multiple runs, each assigned a specific threshold parameter, denoted as $\tau$, given to them as their input parameter. This threshold is a critical input used to assess the usefulness of the elements and is employed as a measure to distinguish between valuable and insignificant elements that are no longer relevant. In each time frame, the run with threshold $\tau$ that meets some specific criteria will have the appropriate solution for that time frame.

A pivotal aspect of our algorithm is the use of a data structure in each run for keeping the elements, from which we can easily retrieve its solution, and it can efficiently be updated.

The fundamental idea behind this data structure is its hierarchical structure. This structure comprises different levels, where each level $\ell$ includes the sets $L_{\ell}$ and $G_{\ell}$. The family of sets $G_{\ell}$ is used to retrieve the solution. Each of them stores the elements that have been selected by the algorithm up to that level, and these sets form a cumulative hierarchy. The set $L_{\ell}$ encompasses the elements with a marginal density of at least $\tau$ with respect to $G_{\ell - 1}$.

Notably, the reconstruction of the entire data structure is a significant operation with potential query complexity implications and one of the key design features of a leveled data structure is its partial reconstruction capability. To further explain, throughout the execution of the algorithm, we can partially reconstruct the data structure starting from the level of our choosing without affecting the previous levels. This feature enables us to handle the insertion or deletion of an element with minimal changes to most levels of our data structure.

Note that even partial reconstructions are heavy operations, and it is in our best interest to avoid them as long as possible, which is why we utilize partial reconstruction only after reaching a level that is heavily affected by the updates up to that point in time and its reconstruction is necessary. To achieve this, in addition to the sets $L_{\ell}$ and $G_{\ell}$, we also maintain a set $D$ and extended sets $\Lp_{\ell}$ to keep track of the inserted and deleted elements, triggering reconstruction when deemed necessary.

To further clarify, when an element is deleted instead of removing that element from the sets $L_{\ell}$ and $G_{\ell}$, we just add it to the set $D$, and when an element is inserted we add it to the set $\Lp_{\ell}$ without changing the sets $L_{\ell}$. While iterating through the levels to make these changes, the sets of each level get inspected and a reconstruction starting from that level is triggered if certain criteria are met. 
It should be noted that a set $L_{\ell}$ gets updated whenever and only when its level is being reconstructed. 

During the formation of levels, or to be more exact when elements from $L_{\ell}$ are being selected for inclusion in $G_{\ell}$, the elements of $L_{\ell}$ are grouped into different buckets, so the elements from the same bucket are approximately similar in aspects of their marginal gain, weight, and marginal density. Then, the largest bucket, denoted as $B_{\ell}$, is chosen, and a suitable number $m_{\ell}$ for the sample size is determined based on the chosen bucket. Consequently, a uniformly random subset of size $m_{\ell}$ gets chosen from the previously mentioned largest bucket $B_{\ell}$ to form the samples $S_{\ell}$. We then form $G_{\ell}$ by adding elements of $S_{\ell}$ to $G_{{\ell}-1}$ one by one if they meet our marginal density requirement. We then remove all elements $e$ with $d(e | G_{\ell}) \le \tau$ from $L_{\ell}$ to form $L_{\ell + 1}$. 

Now, we proceed to explain more about what we meant by a suitable number for sample size. 

Choosing a smaller sample size ensures a larger fraction of elements from $S_{\ell}$ appear in $G_{\ell}$, reducing the impact of deletions of elements of $B_{\ell}$ on the marginal gain of level ${\ell}$, which leads to less need in invoking reconstruction starting from level ${\ell}$. Conversely, a larger sample size may lead to more substantial removals in the filtering step, impacting the number of the levels of the data structure leading to less query complexity for each reconstruction. This is why we use simulation to obtain a sample size to strike a balance and end up with a low overall query complexity.

\subsection{Parallel Runs}
\label{Parallel}

We keep parallel runs and designate the threshold $\tau = (1+\epsilon)^i$ to the run $i$. 

In the section \ref{dyn:approx}, we guarantee that at each point of time, the output of the run with threshold $\tau$, where $\tau \leq \frac{f(V) \cdot  \epsilon }{\optcost} < (1 + \epsilon) \tau$ is an appropriate bicriteria approximation of the solution in that time frame.

We know that $\tau \leq \frac{f(V) \cdot  \epsilon }{\optcost} < (1 + \epsilon) \tau$ is equivalent to 
$\log_{1+\epsilon}{(\tau)} \leq \log_{1+\epsilon}{(\frac{f(V) \cdot  \epsilon }{\optcost})} < 1 + \log_{1+\epsilon}{(\tau)}$. Therefore, it is guaranteed that the output of the instance with index $\floor{\log_{1+\epsilon}{(\frac{f(V) \cdot  \epsilon }{\optcost})}}$, has an appropriate solution.

We know that $1 \leq \optcost \leq |V|\rho$, so we have 
$\frac{f(V) \cdot  \epsilon }{|V|\rho} \leq \frac{f(V) \cdot  \epsilon }{\optcost} \leq f(V) \cdot  \epsilon$, which implies 
$\log_{1+\epsilon}{(\frac{f(V) \cdot  \epsilon }{|V|\rho})} \leq
\log_{1+\epsilon}{(\frac{f(V) \cdot  \epsilon }{\optcost})} \leq
\log_{1+\epsilon}{(f(V) \cdot  \epsilon)} $. Therefore, at any time we only need to search through the instances with index in $[
        \floor{\log_{1+\epsilon}{(\frac{f(V) \cdot \epsilon}{|V|\rho})}}, \floor{\log_{1+\epsilon}{(f(V)\cdot \epsilon)}}]$ to find a valid solution. 

It can be observed in our last argument that the guarantee of a proper solution in the run $i$ with $\tau = (1+\epsilon)^i$ when $\tau \leq \frac{f(V) \cdot  \epsilon }{\optcost} < (1 + \epsilon) \tau$ is sufficient for the correctness of our algorithm. This means that for each run $i$ with $\tau = (1+\epsilon)^i$ we only need to guarantee its correctness when 
$\tau \leq \frac{f(V) \cdot  \epsilon }{\optcost} < (1 + \epsilon) \tau$.

Because of the monotonocity of the function $f$, we know that for any $e \in V$, $f(e) \leq f(V)$. We also know that $d(e) \leq f(e)$, which implies $d(e) \leq f(V)$. We also know that $\optcost \leq n\rho$. Therefore, for any $e \in V$, we have $\frac{d(e) \cdot  \epsilon }{n\rho} \leq \frac{f(V) \cdot  \epsilon }{\optcost}$.
Hence, $\frac{f(V) \cdot  \epsilon }{\optcost} < (1 + \epsilon) \tau$ only holds when for any $e \in V$, $\frac{d(e) \cdot  \epsilon }{n\rho} < (1 + \epsilon) \tau$, which is equivalent to $\log_{1+\epsilon}{(\frac{d(e) \cdot  \epsilon }{n\rho})} < i + 1$. This is why an element $e$ can only be considered in runs with  $i \geq \log_{1+\epsilon}{(\frac{d(e) \cdot  \epsilon }{n\rho})}$.

It should also be noted that an element $e$ with $d(e) < \tau$ will be automatically ignored by the algorithm. So we also can consider an element $e$ only in the runs with $\log_{1+\epsilon}{(d(e))} \geq i$.
Therefore, to handle the update of any element $e$, we only need to invoke \update{} in instances within the specified range. 

\begin{algorithm}[h] 
\caption{Parallel Runs} 
\begin{algorithmic}[1]
    \For{$i \in \mathbb{Z}$} 
        \State Let $\mathcal{I}_i$ be the instance of our dynamic algorithm, for which $\tau =(1+\epsilon)^i$.
    \EndFor

    \Function{GlobalUpdate}{$e$}
        \State update$(V)$
        \For{\textbf{each} $
        \log_{1+\epsilon}{\left( \frac{d(e) \cdot \epsilon}{n \rho (1 + \epsilon)} \right)} 
        \leq i \leq \log_{1+\epsilon}{d(e)}$
        } \label{line:paralle_run_insertion_interval}
            \State Invoke $\update(e)$ for instance $\mathcal{I}_i$.
        \EndFor
    \EndFunction

    \Function{SolutionRetrieval}{$ $}
        \State Let $i^*\in [
        \floor{\log_{1+\epsilon}{(\frac{f(V) \cdot \epsilon}{|V|\rho})}}, \floor{\log_{1+\epsilon}{(f(V)\cdot \epsilon)}}]$ be the index of the instance whose corresponding $G_T$ meets the criteria $f(G_T) \ge (1 - O(\epsilon))f(V)$ and its $G_T \backslash D$ has the lowest cost.
        \State \Return $G_T \backslash D$ of $\mathcal{I}_{i^*}$.
    \EndFunction
\end{algorithmic}
\label{alg:cardinality:unknown:opt}
\end{algorithm}

\subsection{Data Structure Construction}
The $\ReconstructF{(i)}$ function iteratively constructs a leveled data structure built upon levels $\ell < i$. It starts by updating sets $L_i$ to include the elements inserted since its last update that have pre-approved marginal density with respect to $G_{i - 1}$ and to exclude 
the elements that have been since deleted. 
It also updates $\Lp_i$ based on the current $L_i$. 

Then, a process begins, where in each step, the set $G$ of the current level gets selected, and then the elements get filtered based on their marginal density with respect to the selected $G$ to form the set $L$ of the subsequent level. 

This process terminates when there are no elements left in a level's set $L$, which is when the algorithms sets $T$ to the index of the last nonempty level.

To select each $G_{i}$, the elements in $L_i$ get processed, and each of them gets assigned to a bucket in a two-dimensional array based on their weight and their marginal density. Then the largest bucket gets selected and will be named $B_i$. The algorithm determines a specific threshold $\tau_i$ for level $i$ and calculates a suitable sample size $m_i$ using the \CalcSampleCountF{} function. It then selects a uniform subset $S_i$ from $B_i$ and adds them to $G_i$ one at a time if they still meet the marginal density condition.

\begin{algorithm}[h]
  \caption{Data Structure Construction}
  \label{alg:offline}
  \begin{algorithmic}[1]
    \Function{\InitF{}}{$V$}
      \State $L_0 \gets V$, \quad $G_0 \gets \emptyset$, \quad $D \gets \emptyset$, \quad $\Lp_0 \gets L_0$
      \State $L_1 \gets \{e \in L_0: \density{e}{G_0} \ge \tau\}$, \quad $\Lp_1 \gets L_1$
      \State $\ReconstructF{}(1)$
    \EndFunction
    \Function{\ReconstructF}{$i$}
      \State $L_i \gets \Lp_i \backslash D$, \quad $\Lp_i \gets L_i$
      \label{line:level_first_line}
      \While{$L_i \ne \emptyset$}\label{line:level_break}
        \For{$e \in L_{i}$}
            \State $j \gets \floor{\log_{1+\epsilon}(\frac{\density{e}{G_{i-1}}}{\tau})}$
            \State $k \gets \floor{\log_{1+\epsilon}
            (w(e))}$
            \State $\buck_{j, k} \gets \buck_{j, k} \cup \{e\}$
        \EndFor
        \State Let $b_{i, 1}$ and $b_{i, 2}$ be the indices of the largest $\buck$
        \State $B_i \gets \buck_{b_{i, 1}, b_{i, 2}}$, \quad $\tau_i \gets (1+\epsilon)^{b_{i, 1}} \cdot \tau$
        \State $m_i \gets \CalcSampleCountF{}(B_i, G_{i-1}, \tau_{i})$
        \State $S_{i}=[ e_{i, 1},\dots,e_{i, m_i} ] \gets$ Uniform subset of size $m_i$ from $B_i$
        \State $G_{i} \gets G_{i-1}$
        \For{$e \in S_i$}
            \If {$d(e | G_{i}) \ge \tau_{i}$}
                \State $G_{i} \gets G_{i} \cup e$
            \EndIf
        \EndFor
        \State $L_{i+1} \gets \{e \in L_i: \density{e}{G_{i}} \ge \tau\}$
        \State $\Lp_{i+1} \gets L_{i+1}$
        \label{line:level_filter}
        \State $i \gets i + 1$
      \EndWhile
      \State $T \gets i - 1$
    \EndFunction
  \end{algorithmic}
\end{algorithm}

\subsection{Insertion }
In the \InsertF{} Function, we manage the insertion of an element $e$, into our dynamic data structure. First, we remove $e$ from the set of deleted elements $D$, indicating its active status. Next, we add $e$ to the extended set $\Lp_0$. Then, we iterate over the levels, starting from level $1$ up to $T + 1$, where $T$ represents the index of the last nonempty level. At each level, we check if the density of $e$ with respect to $G_{i-1}$ is greater than the threshold $\tau$. If so, $e$ is added to the extended set $\Lp_i$. Otherwise, $e$ would not be added to the extended set $\Lp_i$, and we also no longer need to check the subsequent levels, so we terminate the loop. 
We also monitor the size of $\Lp_i$, and if $|\Lp_i|$ exceeds $\frac{3}{2}|L_i|$, we reconstruct the levels starting from Level $i$ and terminate the loop. 

\subsection{ Deletion} 

The \DeleteF{} function handles the removals of the elements. When an element $e$ is deleted, we begin by adding $e$ to the set $D$, which keeps track of the deleted elements. Then, we iterate over all nonempty levels. In each level, we check if the proportion of deleted elements from the bucket $B_i$ used for sampling $S_i$ exceeds the threshold $\epsilon$. If this condition holds, we trigger a reconstruction of the data structure starting from the current level $i$ using the \ReconstructF{} function and then terminate the loop.

\begin{algorithm}[h]
  \caption{Insertion}
  \label{alg:insert}
  \begin{algorithmic}[1]
    \Function{Insert}{$e$}
        \State $D \gets D \backslash \{e\}$
       \State $\Lp_{0} \gets \Lp_0 \cup \{e\}$
      \For{$i\gets 1, \dots, T + 1$}
        \If{$\density{e}{G_{i-1}} < \tau$}
          \State \Break\label{line:insert_break}
        \EndIf
        \State $\Lp_i \gets \Lp_i \cup \{e\}$ \label{line:insert_add}
        \If{$i=T + 1$ or $|\Lp_{i}| \ge \frac{3}{2}\cdot |L_i|$}
          \State \ReconstructF{}(i)\label{line:insert_level}
          \State \Break
        \EndIf
      \EndFor
    \EndFunction
  \end{algorithmic}
\end{algorithm}

\begin{algorithm}[h]
  \caption{Deletion}
  \label{alg:delete}
  \begin{algorithmic}[1]
    \Function{Delete}{$e$}
      \State $D \gets D \cup e$
      \For{$i\gets 1, \dots, T$}
        \If{$|D \cap B_i| \ge \epsDel \cdot |B_i|$}
          \State \ReconstructF{}(i)
          \State \Break
        \EndIf
      \EndFor
    \EndFunction
  \end{algorithmic}
\end{algorithm}

\subsection{Choice of Sample Size}
We know that as we add elements of $S_i$ to $G_i$, and $G_i$ grows larger, the remaining elements in $S_i$ are less likely to satisfy the marginal density requirement for being added to the $G_i$. Thus, intuitively, choosing a smaller sample size ensures that a larger fraction of the elements in $S_i$ appear in $G_i$. Therefore, deleting an $\epsilon$-fraction of $S_i$ would not drastically affect the value of $\Delta(G_{i}|G_{i-1})$ as opposed to the case where only a few elements of $S_i$ appear in $G_i$ and the deletion of those few elements has a significant impact on the marginal value $\Delta(G_{i}|G_{i-1})$. Therefore, having a smaller sample size leads to less invocation of $\ReconstructF{}$ function. 

On the other hand, choosing a larger $m_i$ ensures that a larger number of elements will be removed in the filtering step. This will reduce the number of levels of our data structure, which decreases the query complexity of the $\ReconstructF{}$ function.

To balance this trade-off, we first try to find out if we process all the elements in $B_i$ one by one in a random order, for any $j \in [1, |B_i|]$, what is the probability of the $j^{\text{th}}$ element being added to $G_i$ and denote such probability by $X^*(j)$. Then we choose the largest integer $m_i^*$ such that $X^*(j) \ge 1-\epsilon$ for all $j \le m_i^*$.
This choice ensures that
\begin{enumerate}
  \item In expectation, $(1-\epsilon)$-fraction of the elements of $S_{i}$ are added to $G_i$
  \item In expectation, at least $\epsilon$-fraction of the elements in $|B_i|$ have their marginal gain decreased sufficiently at the end of this level. Formally, these elements either do not appear in $L_{i+1}$, or the index of their bucket decreases.
\end{enumerate}

Yet, since we cannot calculate the exact values of $X(i)$, we estimate
these probabilities by simulating, and we obtain a sample size that satisfies properties similar to properties of $m_i^*$ with a high probability.

Here we provide a formal definition for the notion of \emph{suitable sample size} for each level. It can be verified that the definition is chosen such that a single run of $\CalcSampleCountF{(L', G', \tau')}$ provides us with a suitable sample size with respect to $L', G',$ and $ \tau'$, with a high probability. For a proofs we refer to Lemmas \ref{lm:bound_countF_always} and \ref{lm:bound_countF} in Appendix \ref{app:invar_proofs}, which were used in both query complexity guarantee and approximation guarantee of the algorithm. 

\begin{definition}[Suitable sample size]\label{def:suitable_sample_size}
  Given the values $L', G',$ and $ \tau'$,
  consider a run of $\apprev$ on these values and let $\bX$ be the $|L'|+1$ dimensional random output.
  A number $m^* \le |L'|$ is called a \emph{suitable sample size with respect to $L', G',$ and $\tau'$} if:
  \begin{align*}
      &\Ex{\bX(r)} \ge 1-2\epsilon \text{ for all }r \in [1, m^*] \text{ and } \\&\Ex{\bX(m^* + 1)} \le 1-\frac{\epsilon}{2}.
  \end{align*}
  We use
  $M^*_{i}$ to denote the set of suitable sample sizes for $(B_i, G_{i - 1}, \tau_{i})$.

\end{definition}

\begin{algorithm}[h]
  \caption{\CalcSampleCountF{}}
  \label{alg:sample}
  \begin{algorithmic}[1]
      \Function{\CalcSampleCountF}{$L', G', \tau'$}
      \State $t \gets \ceil{\constt}$ 
      \For{$j \in [1, t]$}
        \State $X_j \gets \apprev{}(L', G', \tau')$ \label{line:define_X_j}
      \EndFor
       \State Let $m'$ be the smallest index $i \in [1, |L'| + 1]$ for which  $\frac{1}{t}(\sum_{j=1}^t X_j(i))< 1 - \epsilon$ \label{line:find_m'}
      \State Return $m' - 1$ \label{line:output_calcsamplecount}
    \EndFunction
    \Function{\apprev}{$L', G', \tau'$}
      \State Let $[e_1,\dots,e_{|L'|}]$ be a random permutation of $L'$ \label{line:random_permuate}
      \State Let $X$ be a $|L'| + 1$ dimensional vector initialized to $0$.
      \State $G'' \gets G'$
      \For{$i=1$ to $|L'|$}
          \If{$d(e_i|G'') \ge \tau'$}
            \State $X(i) \gets 1$
            \State $G'' \gets G'' \cup \{e_i\}$
          \Else 
            \State \Continue
          \EndIf  
      \EndFor
      \State \Return $X$
    \EndFunction
  \end{algorithmic}
\end{algorithm}

\section{Theoretical Analysis}
In this section, we state our main theoretical results.
\begin{theorem}\label{thm:main}
  We have provided an algorithm for dynamic submodular cover problem, where weight of each item is guaranteed to be in the range $[1, \rho]$,
  that maintains an expected $(1 - O(\epsilon), O(\epsilon)^{-1})$-bicriteria approximate solution 
  with an expected $\poly{}(\log(n), \log(\rho), \eps^{-1})$ amortized oracle queries per update. 
\end{theorem}

Note that as mentioned in the statement of the theorem, our algorithm and our proof use the assumption that weight of each element is in the range $[1, \rho]$. However, it can easily be verified that this theorem proves Theorem \ref{contributiontheorem} as dividing the weight of all element by $\min_{v \in \ground} w(v)$ guarantees the assumption while it would not change the ratio between the cost of our proposed solution and the optimal solution. 

\begin{proof}
We prove this theorem using the following theorems \ref {thm:bicapprox} and \ref{thm:bound_num_query_amor}
regarding the approximation guarantees and query complexity guarantee of the dynamic algorithm that we proposed in section \ref{main:algorithm}, respectively.  
\end{proof}

Complete and detailed proof of the next theorem and alongside its references can be found in Appendix \ref{dyn:approx}.

\begin{theorem}\label{thm:bicapprox}
Our algorithm maintains an expected $(1 - O(\epsilon), O(\epsilon)^{-1})$-bicriteria approximation of the solution.
\end{theorem}

\begin{proof}

Consider the run whose assigned threshold parameter satisfies $\tau \leq \frac{f(V) \cdot  \epsilon }{\optcost} < (1 + \epsilon) \tau$. Lemma \ref{lm:bicapprox} guarantees that output of this run is an expected $(1 - O(\epsilon), O(\epsilon)^{-1})$-bicriteria approximate solution.
 As explained in \ref{Parallel} this run is included in the instances with an index between 
$\floor{\log_{1+\epsilon}{(\frac{f(V) \cdot \epsilon}{|V|\rho})}}$ and $\floor{\log_{1+\epsilon}{(f(V)\cdot \epsilon)}}$, and Lemma \ref{lm:approx:nodel} ensures that $G_T$ of this run satisfies the condition of $f(G_T) \geq (1 - \epsilon) f(V)$.
Additionally Lemma \ref{lm:approx:del} ensures that $f(G_T \backslash D)$
of any run with $f(G_T) \geq (1 - \epsilon) f(V)$ is an expected $(1 - \epsilon)$ approximate of $f(V)$. Therefore, by checking all the instances in the specified range with $f(G_T) \geq (1 - \epsilon) f(V)$, and choosing the one with lowest $\cost(G_T \backslash D)$, we are guaranteed to find an expected $(1 - O(\epsilon), O(\epsilon)^{-1})$-bicriteria approximate solution for the problem. 
\end{proof}

A complete proof of the following theorem and its references can be found in detail in Appendix \ref{Query}.

\begin{theorem}
\label{thm:bound_num_query_amor}
  The expected amortized query complexity of our algorithm is at most
  \begin{math}
    \poly(\log(\eta), \frac{1}{\epsilon})
  \end{math}
  per update, where $\eta := \frac{1+\epsilon}{\epsilon} n \rho$.
\end{theorem}
\begin{proof}
We start by bounding the count of "direct" queries, which come from insertions and deletions. Here, we don't count queries made indirectly through \ReconstructF{}. Each insertion or deletion can result in at most $\mO(\bT)$ queries, where $\bT$ is the number of levels during the update. According to Lemma~\ref{lm:num_level_total}, this is capped at $\poly(\log(n), \log(\eta), \frac{1}{\epsilon})$ because $|\Lc_1| \le n$.

Moving on to \enquote{indirect} queries made by \ReconstructF{}, we charge the cost of each $\ReconstructF{}(i)$ call to the updates causing it. If $\ReconstructF{}(i)$ is triggered by an insertion, its cost is charged to $\Lp_{i} \backslash L_i$, and if by a deletion, it's charged to $B_i \cap D$.
Each time $\ReconstructF{}(i)$ is called for some $i$, the expected number of queries is $|\Lc_i| \poly(\log(|\Lc_i|), \log(\eta), \frac{1}{\epsilon})$. However, this cost is spread across at least $\frac{|\Lc_i|}{\poly(\log(\rho), \log(\eta), \frac{1}{\eps})}$ updates due to the reconstruction conditions (The lower bound is chosen considering the reconstruction condition of deletion and size of $B_i$ and will clearly also hold if $\ReconstructF{}(i)$ is triggered by an insertion, because in that case $|\Lp_{i} \backslash L_i|$ is at least $\frac{1}{3}|\Lp_i| \geq \frac{1}{3}|\Lc_i|$). Hence, the cost of each charge is at most $\poly(\log(\eta),\frac{1}{\eps})$ (note that $|\Lc_i| \leq n$, and $\eta$ has both $\rho$ and $n$ as factors). Now, since each update can be charged by  $\ReconstructF{}(i)$ only once and only if when the update happens $T > i$ and level $i$ gets affected, we can say each update is charged at most $\poly(\log(\eta), \frac{1}{\eps})$ for each of the levels it affects. And as the expected number of levels during the update is at most $\poly(\log(\eta), \frac{1}{\eps})$ by Lemma~\ref{lm:num_level_total}, the claim follows. It's important to note that the random bits used to limit the expectation of $\bT$ and the ones used to limit the queries for each reconstruction are separate. Since the value $\bT$ is known at the update time, it relies on the random bits used before the update. In contrast, the number of queries for each $\ReconstructF{}$ depends on random bits used after (or at the time of) the update.
\end{proof}

Note that the proofs of the Theorems \ref{thm:bicapprox} and \ref{thm:bound_num_query_amor} use the invariants introduced in Appendix \ref{Invariants} and proved in Appendix \ref{app:invar_proofs}. 

\section{Conclusion}
In this paper, we explored the dynamic setting of the monotone submodular cover problem. Specifically, we introduced a $(1- O(\epsilon), \mO(\epsilon^{-1}))$-bicriteria approximation algorithm with polylogarithmic query complexity.

For future research directions, a promising avenue is to refine the query complexity to $\poly(\log(k), \epsilon)$ while making it independent of $n$. 

Moreover, the exploration of the non-monotone version of the submodular cover problem in the dynamic setting remains an open challenge.

\section*{Acknowledgements}
Partially supported by DARPA QuICC, ONR MURI 2024 award on Algorithms, Learning, and Game Theory, Army-Research Laboratory (ARL) grant W911NF2410052, NSF AF:Small grants 2218678, 2114269, 2347322

\section*{Impact Statement}
This paper presents work whose goal is to advance the field of Machine Learning. It contributes to the literature on the weighed submodular cover problem with applications in various fields and ML tasks like data summarization and active set selection.  
There are many potential societal consequences of our work, none of which we feel must be specifically highlighted here.

\bibliographystyle{icml2024}
\bibliography{references}

\appendix
\onecolumn

\section{Invariants}
\label{Invariants}
In this section, we are going to introduce the invariants, which will be used in our analysis. These invariants are guaranteed to hold after each insertion or deletion throughout the execution of the algorithm.

We define the function $\filter{}$ as follows: 
\begin{align*}
    \filter{}(L, G, \tau) := \cbr{e \in L: \density{e}{G} \ge \tau}\text{.}
\end{align*}
Note that we use $\hat{L}_{i}$ to denote $\overline{L}_i \backslash D$. Also recall that we use bold letters to denote random variables while non-bold letters are used to denote the value of variables during the execution.

\textbf{Level invariants} \label{def:invariant_cache}

  \begin{itemize}
  
    \item \textbf{Filter invariant}: 
    $\hat{L}_{i} = \FilterF{}(\hat{L}_{i - 1}, G_{i - 1}, \tau)$
    for all $i \in [T + 1]$.
    
    \item \textbf{Subset invariant}: 
    $\overline{L}_{i} \subseteq \overline{L}_{i - 1}$ for all $i \in [T + 1]$.

    \item \textbf{Deviation invariant}: 
    $|B_i \cap D| \le \epsDel |B_i|$
    and $|\overline{L}_i| \le \frac{3}{2} |L_i|$ for all $i \in [T]$.

    \item \textbf{Stopping invariant}: 
    $\hat{L}_{T+1}= \overline{L}_{T+1} = L_{T+1} = \emptyset$ and 
    $\hat{L}_{i}, \overline{L}_{i}, L_{i}  \ne \emptyset$ for any $ i \in [T]$.

  \end{itemize}

Before continuing with the rest of the invariants, we provide a few more definitions.

\begin{definition}[]
 We define the pre-sample history of Level $i$ as
\begin{equation}
    H_i := (\Lp_0, \Lp_1, \dots, \Lp_i, L_0, L_1, \dots, L_i, G_1, \dots, G_{i-1}, m_i).
    \label{eq:def_history}
  \end{equation}
Intuitively, $H_i$ captures the state, or \enquote{history}, of Algorithm~\ref{alg:offline}
before $\bS_i$ is sampled.

Similarly, we define the pre-size history of Level $i$ to be its pre-sample history minus $m_{i}$, i.e.,
  \begin{equation}
    \histPre_{i} = 
    (\Lp_0, \Lp_1, \dots, \Lp_i, L_0, L_1, \dots, L_i,G_1, \dots, G_{i-1}).
    \label{eq:def_pre_history}
  \end{equation}
  We analogously use $\bH_{\ell}$ and $\bhistPre_{\ell}$ to denote the random variables corresponding to these quantities.
\end{definition}

Now we introduce the random invariants of our algorithm.

\textbf{Sampling invariants}\label{def:invariant_samp}
  \begin{itemize}
    \item \textbf{Sample uniformity invariant}: Conditioned on the pre-sample history of Level $i$,
      the sample set
      $\bS_{i}$ is a uniformly random subset of size $m_i$ from $B_i$.

      Formally, for any $i\ge 1$, and any $H_i$ such that
  $\Pr{\bT \ge i, \bH_i = H_i} > 0$,
  \begin{equation}
    \Pr{\bS_i = S | \bT \ge i, \bH_i = H_i} = \frac{1}{|X_i|}\ind{S \in X_i},
    \label{eq:invariant_uniform}
  \end{equation}
  where $|X_i|$ denotes all sequences of length $m_i$ in $B_i$.  
    
    \item \textbf{Sample size invariant}: 
      Conditioned on the pre-size history of Level $i$, $m_{i}$ is a suitable sample size with a high probability.
      Formally,
      for all $\histPre_i$ such that
      $\Pr{\bT \ge i, \bhistPre_{i} = \histPre_{i}} > 0$,
      \begin{equation*}
      \Pr{\bm_i \in M_i^* | \bT \ge i, \bhistPre_{i} = \histPre_{i}} \ge 
        1 - \frac{\epsilon}{n^{10}}.
      \end{equation*}
  \end{itemize}

  We note that while the above two invariants are intuitively evident for a single execution of Algorithm~\ref{alg:offline} (see Claim~\ref{lm:uniform_right_after}), the dynamic algorithm has the potential to modify both $S_i$ and $H_i$ during update processing. 
Therefore, the result is not immediately clear and needs to be formally proved.
Indeed, the proof heavily relies on the fact that the decision to invoke \ReconstructF{}(i) in \InsertF{} and \DeleteF{} procedures are based only on
$\histPre_i$ and not on $S_i$.

We refer to Appendix~\ref{app:invar_proofs}
for proofs of these invariants.

\section{Approximation guarantee}
\label{dyn:approx}

\begin{lemma}\label{lm:bicapprox}
Assuming that $\tau \leq \frac{f(V) \cdot  \epsilon }{\optcost} < (1 + \epsilon) \tau$, the output (i.e., $G_T \backslash D$) is an expected $(1 - O(\epsilon), O(\epsilon)^{-1})$-bicriteria approximation of the solution.
\end{lemma}
\begin{proof} 
We prove this lemma by the combination of the following Lemmas \ref{lm:costapprox} and \ref{lm:fvalapprox}. 
\end{proof}
\begin{lemma}\label{lm:costapprox}
Assuming that $\tau \leq \frac{f(V) \cdot  \epsilon }{\optcost} < (1 + \epsilon) \tau$, the following holds: 
 \begin{align*}
 \cost (G_T \backslash D) < (\frac{1 + \epsilon }{\epsilon})\optcost.
 \end{align*}
\end{lemma}

\begin{proof}
Let's use $e_1, e_2, \dots, e_{|G_T|}$ to denote the elements in $G_T$ based on the order they were added to the solution sets. 

For any level $1 \leq \ell \leq T$, we know that $\tau_{\ell} \geq \tau$, and the marginal density of any element in $G_T$ must have been greater than some $\tau_{\ell}$ at the time of it was added to the solution sets. Therefore, for any $ 1 \leq i \leq |G_T|$, we know that 
$d(e_i | \{e_1, \dots, e_{i - 1}\}) = \frac{\Delta(e_i | \{e_1, \dots, e_{i - 1}\})}{w(e_i)} \geq \tau$, or equivalently $w(e_i) \leq \frac{\Delta(e_i | \{e_1, \dots, e_{i - 1}\})}{\tau}$. The submodularity of the function $f$ ensures that $\Delta(e_i | \{e_1, \dots, e_{i - 1}\}) \leq \Delta(e_i | \{e_1, \dots, e_{i - 1}\} \backslash D)$. Hence, for any $ 1 \leq i \leq |G_T|$, we have $w(e_i) \leq \frac{\Delta(e_i | \{e_1, \dots, e_{i - 1}\} \backslash D)}{\tau}$.
Therefore, 
$\cost (G_T \backslash D) = \sum_{e \in G_T \backslash D}{w(e)} 
\leq \frac{f(G_T \backslash D)}{\tau}$. 
We have assumed that $\frac{1}{(1 + \epsilon) \tau} < \frac{\optcost}{f(V) \cdot \epsilon}$, which implies $\frac{1}{\tau} < \frac{(1 + \epsilon)\optcost}{f(V) \cdot \epsilon}$.
Therefore, we have $\cost (G_T \backslash D) < \frac{(1 + \epsilon) \optcost}{\epsilon} (\frac{f(G_T \backslash D))}{f(V)}) \leq \frac{(1 + \epsilon) \optcost}{\epsilon}$, where the second inequality follows from the monotonocity of the function $f$ and the fact that $(G_T \backslash D) \subseteq V$. 
    
\end{proof}
\begin{lemma}\label{lm:fvalapprox}
Assuming that $\tau \leq \frac{f(V) \cdot  \epsilon }{\optcost} < (1 + \epsilon) \tau$, the following holds: 
 \begin{align*}
 \Ex{f(\bold{G_T} \backslash D)} \ge (1-O(\epsilon)) f(V).
 \end{align*}
\end{lemma}
\begin{proof}
    To prove this theorem, we ignore the deletions at the begining, and investigate the approximations of $G_T$ even though it might include some deleted elements. We provide an upper bound on $f(G_T)$ in Lemma \ref{lm:approx:nodel}, and then we factor in the removal of deleted elements by bounding the effect of the deleted elements on the value of the function in Lemma  \ref{lm:approx:del}.
\end{proof}

\begin{lemma}
    \label{lm:approx:nodel}
    Assuming that $\tau \leq \frac{f(V) \cdot  \epsilon }{\optcost} < (1 + \epsilon) \tau$, we have 
    $f(G_T) \geq (1 - \epsilon) f(V)$.

\end{lemma}
\begin{proof}
    The statement of the Lemma trivially holds if $f(G_T) \geq  f(V)$.
    Thus, we assume that  $f(G_T) < f(V)$. Recall that $f(V) = f(\solution_{\opt})$. By the monotonicity property of function $f$, we have $f(\solution_{\opt} \cup G_T) \geq f(V)$. 
Let's denote $\solution_{opt} \backslash G_T$ by the set $\{ u_1,u_2,\dots,u_d\}$ for some $0 < d \le |V|$. For each $u \in \solution_{opt} \backslash G_T$, we know that $u \in V = \hat{L}_0$ and $ u \notin  \hat{L}_{T+1} = \emptyset$ by stopping invariant (Lemma \ref{lm:final_level_stronger}). 
Therefore, there exists a level $\ell$ such that $e \in \hat{L}_{\ell} $ but $e \notin \hat{L}_{\ell + 1}$. By the filter invariant (Lemma \ref{lm:invariant_filter}) we know that $e \notin filter (\hat{L}_{\ell}, G_{\ell})$. 
Therefore, we know that $d(e|G_{\ell}) < \tau$. The former inequality and the submodularity of the function $f$ alongside with the fact that $G_\ell \subseteq G_{T}$ implies that $d(e|G_T) < \tau$. Hence, for any $i \in [d]$, we have$\frac{\Delta(u_i | G_T)}{w(u_i)} < \tau$. 
By the submodular property of $f$, we have:
$$ f(\solution_{opt} \cup G_T) - f(G_T) 
            \le \sum_{i=1}^d \Delta(u_i | G_T).$$
Therefore, we have:
$$ f(\solution_{opt} \cup G_T) - f(G_T) \le \sum_{i=1}^d \tau \cdot w(u_i) = \tau \sum_{i=1}^d w(u_i) = \tau \cdot \cost(\solution_{opt} \backslash G_T).$$
We know that $ \cost(\solution_{opt} \backslash G_T) \leq \cost(\solution_{opt}) = \optcost$. So 
$$ f(\solution_{opt} \cup G_T) - f(G_T) \leq  \tau \cdot \optcost \leq \frac{f(V) \cdot  \epsilon }{\optcost} \cdot \optcost = f(V) \cdot  \epsilon, $$ 
where the second inequality holds because of the assumption of the Lemma.
Previous equation yields  that $f(G_T) \geq f(\solution_{opt} \cup G_T) - (f(V) \cdot \epsilon) \geq f(V) - (f(V) \cdot \epsilon) = (1 - \epsilon) f(V).$
\end{proof}

Now, in the following lemma, we give an upper bound on the effect of the deleted elements and thus prove an expected lower bound for $f(\bold{G_T}\backslash D)$. 

\begin{lemma}
  \label{lm:approx:del}
  $\Ex{f(\bold{G_T} \backslash D)} \ge (1-O(\epsilon)) \Ex{f(\bold{G_T})}$
\end{lemma}
\begin{proof}
  Define the variables $g_{\ell}$ and $l_{\ell}$ to denote the gain and loss of each level $\ell \le T$ as:
  \begin{align*}
      {g}_{\ell} := {\Delta}(({G}_{\ell} \backslash {G}_{\ell - 1})|{G}_{\ell - 1}), \quad {l}_{\ell} := {\Delta}((({G}_{\ell} \backslash {G}_{\ell - 1}) \cap D) |{G}_{\ell - 1}).
  \end{align*}
  For $\ell \ge T + 1$, define ${g}_{\ell} = {l}_{\ell} = 0$.
  
  we need to show that
  \begin{align*}
      \Ex{\sum_{\ell=1}^{\infty} (\bold{g_{\ell}} - \bold{l_{\ell})}} \ge (1-O(\epsilon)) \Ex{\sum_{\ell=1}^{\infty} \bold{g_{\ell}}}.
  \end{align*}
  Note that the summation is over all $\ell \ge 1$ and $g_{\ell} = l_{\ell} = 0$ for $\ell > T$.
  The above inequality is equivalent to
  \begin{align}
      \Ex{\bold{\sum_{\ell=1}^{\infty} l_{\ell}}} \le O(\epsilon) \Ex{\bold{\sum_{\ell=1}^{\infty} g_{\ell}}}
      \label{eq:expec_l_expect_w}
  \end{align}
  We will show that
  \begin{math}
      \Ex{\bold{l_{\ell}}} \le  O(\epsilon) \Ex{\bold{g_{\ell}}}
  \end{math} 
  holds for all $\ell\ge 1$. Summing over $\ell$, we obtain the equation \eqref{eq:expec_l_expect_w} by linearity of expectation.
  
  By the law of total expectation, it suffices to prove
  \begin{align}
    \Ex{\bold{l_{\ell}} | \bhistPre_\ell = \histPre_{\ell}} \le  O(\epsilon) \Ex{\bold{g_{\ell}} | \bhistPre_\ell = \histPre_{\ell}},
    \label{eq:nov1_1035}
  \end{align} 
  for all $\histPre_{\ell}$ such that $\Pr{\bhistPre_\ell = \histPre_{\ell}} > 0$.
  To do this, we will first show that the claim holds if $\bm_\ell \in M_\ell^*$. 
  Given the Sample size invariant (Lemma \ref{lm:bound_countF_always}),
  $\bm_\ell \in M_\ell^*$ holds with high probability, which will later allow us to prove Equation \eqref{eq:nov1_1035}.
  
  Consider any value 
  $m_{\ell} \in M_{\ell}^*$ and
  define $H_{\ell}$ as $(\histPre_{\ell}, m_{\ell})$.
  Assuming that
  $H_{\ell}$ is such that $\Pr{\bT \ge \ell, \bH_{\ell} = H_{\ell}} > 0$, we can claim that
  \begin{align}
    \Ex{\bold{l_{\ell}} | \bH_{\ell}= H_{\ell}} \le O(\epsilon) \Ex{\bold{g_{\ell}} | \bH_{\ell} = H_{\ell}}.
    \label{eq:nov1_1031}
  \end{align}
    To prove this, we first note that we are taking $m_{\ell}$ samples from $B_{\ell}$, and the weight of all the elements in $B_{\ell}$ is in range $[(1 + \epsilon)^{k}, (1 + \epsilon)^{k + 1} )$ for some integer $0 \leq k \leq \floor{\log_{1 + \epsilon}(\rho)}$. 

To prove this, we first note that we are taking $m_{\ell}$ samples from $B_{\ell}$, and the weight of all the elements in $B_{\ell}$ is in range $[(1 + \epsilon)^{k}, (1 + \epsilon)^{k + 1} )$ for some integer $0 \leq k \leq \floor{\log_{1 + \epsilon}(\rho)}$.

 Lemma \ref{lm:quality_sample_count} implies that
  \begin{align}
    \Ex{\bold{g_{\ell}} | \bH_{\ell} = H_{\ell}} \ge (1-2\epsilon) \cdot m_{\ell} \cdot \tau_{\ell} \cdot (1 + \epsilon)^{k}.
    \label{eq:nov1_1030}
  \end{align}
    Furthermore, considering deviation invariant (Lemma \ref{lm:reconstruction_condition}), $|D \cap B_{\ell} | \le \epsilon \cdot |B_{\ell}|$, and given the sample uniformity invariant (Lemma \ref{lm:uniform_lazy}), which states $\bold{S_{\ell}}$ is a uniformly random subset of $|B_{\ell}|$ of size $m_{\ell}$, we have:
  \begin{align*}
    \Ex{|D \cap \bold{S_{\ell}}|  \mid \bH_{\ell} = H_\ell} \le \epsilon \cdot m_{\ell}.
  \end{align*}
  
  Since $G_{\ell} \backslash G_{\ell-1}$ is a subset of $S_{\ell}$, we also have 
  $\Ex{|D \cap (G_{\ell} \backslash G_{\ell-1})|} \le \epsilon \cdot m_{\ell}$.
  Since marginal density of each element in $G_{\ell} \backslash G_{\ell-1}$ is at most $\tau_{\ell} \cdot (1+\epsilon)$, which implies marginal value of each element is at most $\tau_{\ell} \cdot (1+\epsilon)\cdot (1 + \epsilon)^{k+1}$, we have
  \begin{align*}
    \Ex{\bold{l_{\ell}} | \bH_{\ell} =H_{\ell}} \le \epsilon(1+\epsilon) \cdot m_{\ell} \cdot \tau_{\ell} \cdot (1 + \epsilon)^{k+1}
    ,
  \end{align*}
  which together with Equation \eqref{eq:nov1_1030}, proves Equation \eqref{eq:nov1_1031}. Equation \eqref{eq:nov1_1031} and law of total expectation imply that
    \begin{align}
    \Ex{\bold{l_{\ell}} | \bhistPre_\ell= \histPre_{\ell} , \bm \in M_\ell^* } \le  O(\epsilon) \Ex{\bold{g_{\ell}} | \bhistPre_\ell = \histPre_{\ell} , \bm \in M_\ell^* }.
    \label{eq:nov1_1040}
  \end{align}

    We have:
  \begin{align}
      \Ex{\bold{g_{\ell}} | \bhistPre_{\ell} = \histPre_{\ell}} &\ge 
      \Pr{\bm_\ell \in M_\ell^* | \bhistPre_{\ell}=\histPre_{\ell}}
      \cdot \Ex{\bold{g_{\ell}} | \bhistPre_{\ell}=\histPre_{\ell}, \bm_\ell \in M_\ell^*} 
      \\&\ge (1-O(\epsilon)) \cdot \Ex{\bold{g_{\ell}} | \bhistPre_{\ell}=\histPre_{\ell}, \bm_\ell \in M_\ell^*} 
      \\&\ge \frac{1}{2} \Ex{\bold{g_{\ell}} | \bhistPre_{\ell}=\histPre_{\ell}, \bm_\ell \in M_\ell^*} 
    \label{eq:nov1_1041}
  \end{align}

On the other hand, we have:
\begin{align*}
      \Ex{\bold{l_{\ell}} | \bhistPre_{\ell} = \histPre_{\ell}} &= 
      \Pr{\bm_\ell \in M_\ell^* | \bhistPre_{\ell}=\histPre_{\ell}}
      \cdot \Ex{\bold{l_{\ell}} | \bhistPre_{\ell}=\histPre_{\ell}, \bm_\ell \in M_\ell^*} \\& + 
      \Pr{\bm_\ell \notin M_\ell^* | \bhistPre_{\ell}=\histPre_{\ell}}
      \cdot \Ex{\bold{l_{\ell}} | \bhistPre_{\ell}=\histPre_{\ell}, \bm_\ell \notin M_\ell^*} \\& \le 
      1 \cdot \Ex{\bold{l_{\ell}} | \bhistPre_{\ell}=\histPre_{\ell}, \bm_\ell \in M_\ell^*}
      +
        \mO(\frac{\epsSamp}{n^{10}}) \cdot \Ex{\bold{l_{\ell}} | \bhistPre_{\ell}=\histPre_{\ell}, \bm_\ell \notin M_\ell^*}
         \\& \le 
         O(\epsilon) \Ex{\bold{g_{\ell}} | \bhistPre_\ell = \histPre_{\ell} , \bm_\ell \in M_\ell^* }
      +
      \mO(\frac{\epsSamp}{n^{10}}) \cdot \Ex{\bold{l_{\ell}} | \bhistPre_{\ell}=\histPre_{\ell}, \bm_\ell \notin M_\ell^*}
        \\& \le 
       O(\epsilon) \Ex{\bold{g_{\ell}} | \bhistPre_\ell = \histPre_{\ell} , \bm_\ell \in M_\ell^* }
      +
     \mO(\frac{\epsSamp}{n^{10}}) \Ex{\bold{g_{\ell}} | \bhistPre_{\ell}=\histPre_{\ell}, \bm_\ell \notin M_\ell^*},
\end{align*}
where the first inequality comes from sample size invariant,  the second equality comes  Equation \eqref{eq:nov1_1040}, and the third one comes from the fact that $\bl_{\ell} \leq \bg_{\ell}$.

The previous equation and the Equation \eqref{eq:nov1_1041}
imply that 
\begin{align}
      \Ex{\bold{l_{\ell}} | \bhistPre_{\ell} = \histPre_{\ell}} & \leq
      O(\epsilon) \Ex{\bold{g_{\ell}} | \bhistPre_{\ell} = \histPre_{\ell}} +
     \mO(\frac{\epsSamp}{n^{10}}) \Ex{\bold{g_{\ell}} | \bhistPre_{\ell}=\histPre_{\ell}, \bm_\ell \notin M_\ell^*}.
     \label{eq:993646}
\end{align}

From the definition, we know that $\Ex{\bold{g_{\ell}} | \bhistPre_{\ell}=\histPre_{\ell}, \bm_\ell \notin M_\ell^*} \leq 
\bm_\ell \cdot \max_{e \in B_{\ell}}{ \{\Delta(e|{G}_{\ell - 1}) \}}$. Also, $\bm_\ell $ clearly is not greater than $n$. Hence, $\Ex{\bold{g_{\ell}} | \bhistPre_{\ell}=\histPre_{\ell}, \bm_\ell \notin M_\ell^*} \le n\cdot \max_{e \in B_{\ell}}{ \{\Delta(e|{G}_{\ell - 1}) \}}$. 

Additionally, we know that the first sampled element from $B_\ell$ always gets added to ${G}_{\ell}$ as it definitely meets the threshold requirement. Therefore, we know that $\Ex{\bold{g_{\ell}} | \bhistPre_{\ell}= \histPre_{\ell}}$ is at least equal to the expected marginal gain of a random element in $B_\ell$, which is definitely at least $\frac{1}{|B_\ell|}(\max_{e \in B_{\ell}}{ \{\Delta(e|{G}_{\ell - 1}) \}})$. $|B_\ell|$ is also at most $n$, so $\max_{e \in B_{\ell}}{ \{\Delta(e|{G}_{\ell - 1}) \}} \leq n \cdot \Ex{\bold{g_{\ell}} | \bhistPre_{\ell}= \histPre_{\ell}}$. Therefore, we have $\Ex{\bold{g_{\ell}} | \bhistPre_{\ell}=\histPre_{\ell}, \bm_\ell \notin M_\ell^*} \le n^2 \cdot \Ex{\bold{g_{\ell}} | \bhistPre_{\ell}= \histPre_{\ell}}$. This last inequality and Equation \eqref{eq:993646} prove Equation \eqref{eq:nov1_1035} and our proof is complete.

\end{proof}

\begin{lemma}\label{lm:quality_sample_count}
  Consider a value of $H_{\ell}$
  for the history up to cache $\ell$
  such that $\Pr{\bH_{\ell} = H_{\ell}} > 0$, and  
  define $g(A)$ as $f(A|G_{\ell - 1})$.Assume that the weights of all the elements in $B_{\ell}$ are in range $[(1 + \epsilon)^{k}, (1 + \epsilon)^{k + 1} )$.
  If $m_{\ell} \in \bM_\ell^*$, then 
  \begin{align*}
    \Ex{\bold{g}(\bold{G}_{\ell} \backslash G_{\ell-1}) | \bH_{\ell}=H_{\ell}} \ge (1-2\epsilon) \cdot \tau_{\ell} \cdot m_{\ell}.
  \end{align*}
\end{lemma}

\begin{proof}
    By definition of suitable sample size, we know that
  for all $i\in [m_{\ell}]$, $\Ex{\bX(i)} \ge 1-2\epsilon$. So for each $i\in [m_{\ell}]$, the $i^{\text{th}}$ element of $S_\ell$ is added to the $G_\ell$ with a probability of at least $1-2\epsilon$.
    We also know that each element is added only if its marginal density is at least $\tau_{\ell}$, which means
    its marginal value is at least $\tau_{\ell}(1 + \epsilon)^{k}$, and increases the value of $f(G_\ell)$ by at least $\tau_{\ell}(1 + \epsilon)^{k}$. Therefore,
    \begin{align*}
    \Ex{\bold{g}(\bold{G}_{\ell} \backslash G_{\ell-1}) | \bH_{\ell}=H_{\ell}} \ge (1-2\epsilon) \cdot \tau_{\ell} \cdot m_{\ell} \cdot (1 + \epsilon)^{k}
  \end{align*}
  as claimed.
\end{proof}

\subsection{Proof of Theorem \ref{thm:bicapprox}}

Consider the run whose assigned threshold parameter satisfies $\tau \leq \frac{f(V) \cdot  \epsilon }{\optcost} < (1 + \epsilon) \tau$. Lemma \ref{lm:bicapprox} guarantees that output of this run is an expected $(1 - O(\epsilon), O(\epsilon)^{-1})$-bicriteria approximate solution.
 As explained in \ref{Parallel} this run is included in the instances with an index between 
$\floor{\log_{1+\epsilon}{(\frac{f(V) \cdot \epsilon}{|V|\rho})}}$ and $\floor{\log_{1+\epsilon}{(f(V)\cdot \epsilon)}}$, and Lemma \ref{lm:approx:nodel} ensures that $G_T$ of this run satisfies the condition of $f(G_T) \geq (1 - \epsilon) f(V)$.
Additionally Lemma \ref{lm:approx:del} ensures that $f(G_T \backslash D)$
of any run with $f(G_T) \geq (1 - \epsilon) f(V)$ is an expected $(1 - \epsilon)$ approximate of $f(V)$. Therefore, by checking all the instances in the specified range with $f(G_T) \geq (1 - \epsilon) f(V)$, and choosing the one with lowest $\cost(G_T \backslash D)$, we are guaranteed to find an expected $(1 - O(\epsilon), O(\epsilon)^{-1})$-bicriteria approximate solution for the problem.

\section{Query complexity}
\label{Query}
In this section, we provide an analysis of the query complexity for our algorithm.
The proofs in this section follow the framework of \cite{banihashem2023dynamic} and are provided for completeness. 

Throughout the section, we set $\eta := \frac{1+\epsilon}{\epsilon} n \rho$. We also use the shorthand
$\sigma_i$ and $\sigmaPre_i$ to denote the events
$\rbr{\bT \ge i \land \bH_i = H_i}$ and $\rbr{\bT \ge i \land \bhistPre_i = \histPre_i}$ respectively.

We begin with the following definition. 
\begin{definition}[Potential]
  For any $i \ge 1$ and any element $e \in \Lc_i$, 
  \emph{potential of $e$ in level $i$}, 
  is defined as the number
  $P(e, i)$ satisfying
  $\frac{\density{e}{G_{i-1}}}{\tau} \in [(1+\epsBuck)^{P(e, i)-1}, (1+\epsBuck)^{P(e, i)})$. 
  For elements $e\notin \Lc_i$, we define
  $P(e, i) = 0$.
  We define the \emph{potential of level $i$}
  as
  $P_i := \sum_{e \in V} P(e, i) = \sum_{e\in \Lc_i} P(e, i)$
  for $i \le T+1$ and $P_{i} = 0$ for $i > T+1$.
\end{definition}

\begin{lemma}\label{lm:potential_properties}
    For any $e$ and $i$ the following hold:
    \begin{enumerate}
        \item $P(e, i) \in [1, O(\log_{1+\epsilon}(\eta))]$.
        \item $\abs{\hat{L}_i} \le P_i \le O(\log_{1+\epsilon}(\eta)) \abs{\hat{L}_i}$
        \item $P_{T+1} = 0$ and $P_i > 0$ for $i \le T$. 
    \end{enumerate}
\end{lemma}
\begin{proof}
    For the first result, note that since
    $e \in \Lc_i$, given Lemma~\ref{lm:invariant_filter},
    we have 
    $e \in \FilterF{}(\Lc_{i-1}, G_{i-1}, \tau)$, which implies
    $\density{e}{G_{i-1}} \ge \tau$. Therefore, $P(e, i) \ge 1$. 
    On the other hand, based on Line~\ref{line:paralle_run_insertion_interval}, we only add elements that satisfy $d(e) \le \eta \tau$. 
    By submodularity, this implies that $\density{e}{G_{i-1}} \le \eta \tau$, thus proving the claim.
    The second result follows from the first one since
    $P(e, i) = 0$ for $e\notin \Lc_i$
    and the third result follows from the second one and Lemma \ref{lm:final_level_stronger}. 
\end{proof}

\begin{lemma}\label{lm:potential_decrease_always}
    For any $i \in [1, T]$ and any $e \in V$, the following inequalities hold:
    $P(e, i) \ge P(e, i+1)$ and, consequently, $P_{i} \ge P_{i+1}$.
\end{lemma}

\begin{proof}
  If $e \notin \Lc_{i+1}$, the claim holds trivially because the right-hand side equals zero.
  Otherwise, since $\Lc_{i+1} \subseteq \Lc_{i}$ by Lemma~\ref{lm:invariant_filter}, and $e\in \Lc_i$, the claim follows from the fact that $G_{i} \subseteq G_{i+1}$. 
\end{proof}

\begin{lemma}\label{lm:potential_decrease}
  Assume we are given sets $L_{j}, G_{j-1}$
  satisfying
  $\FilterF{}(L_j, G_{j-1}, \tau) = L_j$,
  and we invoke $\ReconstructF{}(j)$
  obtaining the (random) values
  $\bT, \bS_{j}, \dots \bS_{\bT}, \bL_{j + 1}, \dots \bL_{\bT+1}$.
  Then,
  for each $i\ge j$,
  \begin{equation*}
    \Ex{P_{i} - \bP_{i+1} \big| \bT \ge i, \bhistPre_i=\histPre_i} \ge \Omega(\epsSamp) \cdot |B_i|
  \end{equation*}
  for all $\histPre_i$ such that $\Pr{\bT \ge i, \bhistPre_i=\histPre_i} > 0$,
  where $\histPre_i$ is defined as in \eqref{eq:def_pre_history}.
\end{lemma}
\begin{proof}
  We first observe that since we are considering
  these values right after we invoke $\ReconstructF{}(j)$, we have $\Lc_i = L_i$. 
  We first give a sketch of the proof.
  For each $i\ge j$,
  by Lemma \ref{lm:bound_countF},
  $\bm_i \in M_i^*$
  with probability at least $1-O(\epsSamp/n^{10})$.
  By definition of $M_i^*$, $\bm_i \in M_i^*$ means that
  in expectation,
  at least
  $\epsSamp/2$ fraction of the elements in $B_i$ satisfy $P_{i+1}(b, e) \le P_i(b, e) - 1$.
  As for the case $\bm_i \notin M_i^*$, 
  since this only happens with low probability, we can handle it using the bound
  $P_{i+1} \le P_{i}$.

  Formally, using the shorthand $\sigmaPre_i$ instead of $\bT \ge i, \bhistPre_i=\histPre_i$,
  \begin{align*}
    &\Ex{P_{i} - \bP_{i+1} \big| \sigmaPre_i}
    \\&=
    \Pr{\bm_i \in M_i^* | \sigmaPre_i} \cdot \Ex{P_{i} - \bP_{i+1} \big| \sigmaPre_i, \bm_i \in M_i^*}
    +
    \Pr{\bm_i \notin M_i^* | \sigmaPre_i} \cdot \Ex{P_{i} - \bP_{i+1} \big| \sigmaPre_i, \bm_i \notin M_i^*}
    \\&\ge
    \Pr{\bm_i \in M_i^* | \sigmaPre_i} \cdot \Ex{P_{i} - \bP_{i+1} \big| \sigmaPre_i, \bm_i \in M_i^*}
  \end{align*}
  where the inequality follows from the fact that $P_i - P_{i+1}$ is always non-negative given Lemma~\ref{lm:potential_decrease_always}.
  By Lemma \ref{lm:bound_countF},
  we can further bound this as
  \begin{align}
    \notag
    \Ex{P_{i} - \bP_{i+1} \big| \sigmaPre_i}
    &\ge
    \notag
    \Pr{\bm_i \in M_i^* | \sigmaPre_i} \cdot \Ex{P_{i} - \bP_{i+1} \big| \sigmaPre_i, \bm_i \in M_i^*}
    \notag
    \\&\ge
    \notag
    (1-\mO(\frac{\epsSamp}{n^{10}})) \cdot \Ex{P_{i} - \bP_{i+1} \big| \sigmaPre_i, \bm_i \in M_i^*}
    \\&\ge
    \notag
    \frac{1}{2} \cdot \Ex{P_{i} - \bP_{i+1} \big| \sigmaPre_i, \bm_i \in M_i^*}
    \\&=
    \frac{1}{2} \cdot \Exu{m_i \sim \bm_i | \sigmaPre_i, \bm_i \in M_i^*}{\Ex{P_{i} - \bP_{i+1} \big| \sigmaPre_i, \bm_i = m_i}}
    \label{yq8wer}
  \end{align}
  Finally, as $\Lc_i = L_i$, we observe that
  \begin{align*}
    P_i - P_{i+1} 
    =
    \sum_{e\in L_i} P(e, i) - \sum_{e\in L_{i+1}} P(e, i+1)
    &=
    \sum_{e\in L_i} P(e, i) - \sum_{e\in L_{i}} P(e, i+1)
    \\&=
    \sum_{e \in B_i}
     \rbr{P(e, i) - P(e, i+1)}
     + \sum_{e \in L_i \backslash B_i}
     \rbr{P(e, i) - P(e, i+1)}
    \\&\overset{(a)}{\ge}
    \sum_{e \in B_i}
     \rbr{P(e, i) - P(e, i+1)}
     \\&\ge
    \sum_{e \in B_i}
    \ind{P(e, i) - P(e, i+1) \ge 1}
     \\&\overset{(b)}{{=}}
    \sum_{e \in B_i} 
    \ind{\density{e}{G_i} < \tau_i}
   \\&\overset{(c)}{=}
     \abs{B_i} - 
    \abs{\FilterF{}(B_i, G_i, \tau_i) }
  \end{align*}
  In the above derivation, inequality
  $(a)$ follows from Lemma~\ref{lm:potential_decrease_always},
   $(b)$ follows from the fact that if $\density{e}{G_i} < \tau_i$ then $P(e, i + 1)<P(e, i)$ will hold for $e \in B_i$ since $\density{e}{G_{i-1}} \ge \tau_i$ for these elements. 
  Finally, $(c)$ follows from the definition of $\FilterF{}$. 
  
  Therefore, 
  \begin{align*}
    \Ex{P_{i} - \bP_{i+1} \big| \sigmaPre_i}
    &\ge
    \frac{1}{2} \cdot \Exu{m_i \sim \bm_i | \sigmaPre_i, \bm_i \in M_i^*}{\Ex{P_{i} - \bP_{i+1} \big| \sigmaPre_i, \bm_i = m_i}}
    \\&\ge
    \frac{1}{2} \cdot \frac{\epsSamp}{2} |B_i|
  \end{align*}
  where the first inequality follows from \eqref{yq8wer} and
  the final inequality follows from Definition~\ref{def:suitable_sample_size}.
\end{proof}

Next, we state the following inequality and refer to
\cite{banihashem2023dynamicmat} for a proof.
\begin{lemma}[Lemma 25 in \cite{banihashem2023dynamicmat}]\label{lm:auxil}
    Let $\bX_0, \bX_1, \dots, \bX_n$ be a sequence of integer positive variables such that $\bX_{i} \le \bX_{i-1}$ and 
    \begin{align*}
        \Ex{\bX_i \mid \bX_1 = X_1, \dots, \bX_{i-1} = X_{i-1}}
        \le 
        (1-\eps') {X_{i-1}}. 
    \end{align*}
    Let $T$ denote the first index $i$ such that $X_T = 0$ and assume that $\bX_0=N$ for some fixed integer $N$. Then
    $\Ex{\bT} \le \frac{\log(N) + 1}{\poly(\eps')}$.
\end{lemma}

\begin{lemma}\label{lm:num_level_reconstruct}
  Assume we are given sets $L_{j}, G_{j-1}$
  satisfying
  $\FilterF{}(L_j, G_{j-1}, \tau) = L_j$ and $L_j = \Lc_{j}$.
  Let 
  $\bT, \bS_{j}, \dots \bS_{T}, \bL_{j + 1}, \dots \bL_{T+1}$
  denote the values after invoking
  $\ReconstructF{}(j)$,
  the expected value of $\bT - j$ is bounded by
  $\poly(\log(|L_j|), \log(\eta), \frac{1}{\eps})$.
\end{lemma}
\begin{proof}
  By Lemma~\ref{lm:potential_properties}
  \begin{align*}
    |L_i| \le P_i \le |L_i| \cdot \log_{1+\epsBuck}(\eta)
  \end{align*}
  Given Lemma \ref{lm:potential_decrease},
  for $i\ge j$,
  \begin{align*}
    \Ex{P_{i} - \bP_{i+1} | \sigmaPre_i}
    \ge \Omega(\epsSamp) \cdot |B_i|
    &\overset{(a)}{\ge} \Omega(\frac{\epsSamp}{\log_{1+\epsBuck}(\eta)\log_{1+\epsBuck}(\rho)}) \cdot |L_i|
    \\&\overset{(b)}{\ge} \Omega(\frac{\epsSamp}{\log^2_{1+\epsBuck}(\eta)\log_{1+\epsBuck}(\rho)}) \cdot |P_i|
    \\&\overset{(c)}{\ge} \Omega(\frac{\epsilon^4}{\log^3(\eta)}) P_i
  \end{align*}
  where for $(a)$, we have used the fact that $|B_i|$ was
  the largest bucket, for $(b)$ we have used Lemma~\ref{lm:potential_properties},
  and for $(c)$, we have used the inequality
  $\log(1+x) \ge \frac{x}{4}$ for $x < 1$.
  Therefore, for some 
  $\eps' = \Omega(\frac{\epsilon^4}{\log^3(\eta)})$,
   we have
  \begin{align*}
      \Ex{\bP_{i+1} | \sigmaPre_i 
      }
      \le 
      (1-\eps')\cdot P_{i},
  \end{align*}
  Since the values $\bP_1, \dots, \bP_i$ are deterministic conditioned on $\sigmaPre_i$, this further implies
  \begin{align}
      \Ex{\bP_{i+1} | 
      \bP_{i}=P_i,
      \dots, \bP_1 = P_1
      }
      \le 
      (1-\eps')\cdot P_{i}.
      \label{qjhwq78}
  \end{align}
  Formally, since when $P_{i} \ne 0$, we have $\bT \ge i$ given Lemma~\ref{lm:potential_properties}, and followed by iterated expectation, we will have
  \begin{align*}
      \Ex{\bP_{i+1} | 
      \bP_{i}=P_i,
      \dots, \bP_1 = P_1
      }
      &=
      \Ex{\bP_{i+1} | 
      \bP_{i}=P_i,
      \dots, \bP_1 = P_1, 
      \bT \ge i
      }
      \\&=
      \Ex{
      \Ex{\bP_{i+1} | 
      \bP_{i}=P_i,
      \dots, \bP_1 = P_1, 
      \bT \ge i,
      \bH_i = H_i
      }
      }
      \\&=
      \Ex{
      \Ex{\bP_{i+1} | 
      \bT \ge i,
      \bH_i = H_i
      }
      }
      \\&\le (1-\eps')P_i
  \end{align*}
  where the third equality follows from the
  fact that 
  $\bP_1, \dots, \bP_i$ is deterministic conditioned on $\cbr{\bT \ge i, \bH_i = H_i}$.
  If $P_i = 0$, then
  \eqref{qjhwq78} holds trivially because
  $P_{i+1}\le P_{i} = 0$. 
  The claim now follows from Lemma~\ref{lm:auxil}
  and the fact that
  $P_j \le |L_j| O(\log(\eta) / \epsBuck)$
\end{proof}

\begin{lemma}\label{lm:bound_num_query_recounstruct}
  The expected number of queries made by calling \ReconstructF{}(i) is
  $|\Lc_i| \cdot \poly(\log(|\Lc_i|), \log(\eta), \frac{1}{\eps}) $, where $|\Lc_i|$ refers to the size of $\Lc_i$ after the update.
\end{lemma}
\begin{proof}
    By Lemma \ref{lm:bound_countF},
    and the fact that
     $\FilterF{}$ and bucketing make at most $\mO(|\Lc_i|)$ queries, 
    each iteration the while-loop in algorithm \ref{alg:offline}
    makes at most
    $\mO((|\Lc_i|) \cdot \poly(\log(\eta), \frac{1}{\epsilon}))$ queries.
    By Lemma~\ref{lm:num_level_reconstruct}, in expectation, the while-loop is executed at most 
  $\poly(\log(|\Lc_i|), \log(\eta), \frac{1}{\eps})$ times, which proves the claim.
\end{proof}

\begin{lemma}\label{lm:num_level_total}
  At any point in the stream,
  the expected number of levels $\Ex{\bT}$ is 
  at most
  $\poly(\log(n), \log(\eta), \frac{1}{\eps})$.
\end{lemma}
\begin{proof}
  We note that this lemma
  is different from Lemma \ref{lm:num_level_reconstruct} because we are claiming that the number of levels is always bounded in expectation,
  while Lemma~\ref{lm:num_level_reconstruct} can only bound $\Ex{\bT- j}$
  right after a call to $\ReconstructF{}(j)$.
  In particular, this also means that we can no longer assume $\Lc_i=L_i$. 
  The proof follows using a similar technique as Lemma \ref{lm:num_level_reconstruct}.
  We first claim that
  a variant of Lemma \ref{lm:potential_decrease}
  still holds.
  More formally, 
  we claim that
  \begin{equation}
    \Ex{P_{i} - \bP_{i+1} \big| \bT \ge i, \bhistPre_i=\histPre_i} \ge \Omega(\epsSamp) \cdot |B_i|
    \label{eq:jul12_1744}
  \end{equation}
  for any $i \ge 1$.
  The proof follows with the exact same logic as the proof of Lemma \ref{lm:potential_decrease}.
  Using the shorthand $\sigmaPre_i$ to denote
  $\bT \ge i, \bhistPre_i=\histPre_i$,
  for all $i\ge 1$,
  \begin{align}
    \notag
    \Ex{P_{i} - \bP_{i+1} \big| \sigmaPre_i}
    &\ge
    \notag
    \Pr{\bm_i \in M_i^* | \sigmaPre_i} \cdot \Ex{P_{i} - \bP_{i+1} \big| \sigmaPre_i, \bm_i \in M_i^*}
    \\&\overset{(a)}{\ge}
    \notag
    (1 - \mO(\frac{\epsSamp}{n^{10}}))\cdot \Ex{P_{i} - \bP_{i+1} \big| \sigmaPre_i, \bm_i \in M_i^*}
    \\&\ge
    \notag
    \frac{1}{2}\cdot \Ex{P_{i} - \bP_{i+1} \big| \sigmaPre_i, \bm_i \in M_i^*}
    \\&=
    \frac{1}{2} \cdot \Exu{m_i \sim \bm_i | \sigmaPre_i, \bm_i \in M_i^*}{\Ex{P_{i} - \bP_{i+1} \big| \sigmaPre_i, \bm_i = m_i}}
    \label{lakjsdf}
  \end{align}
  where for (a) we have now used Lemma \ref{lm:bound_countF_always} instead of \ref{lm:bound_countF}.
  
  Next, we observe that
  \begin{align*}
    P_{i} - P_{i+1}
    &=
    \sum_{e \in \Lc_i} (P(e, i) - P(e, i+1))
    \\&\ge
    \sum_{e \in B_i\cap \Lc_i}
    \ind{P(e, i) - P(e, i+1) \ge 1}
     \\&\ge
    \sum_{e \in B_i\cap \Lc_i}
    \ind{\density{e}{G_i} < \tau_i}
    \\&=
    \sum_{e \in B_i}
    \ind{\density{e}{G_i} < \tau_i}
    - 
    \sum_{e \in B_i \backslash \Lc_i}
    \ind{\density{e}{G_i} < \tau_i}
    \\&\overset{(a)}{\ge}
    \sum_{e \in B_i}
    \ind{\density{e}{G_i} < \tau_i}
    - 
    \sum_{e \in B_i \cap D}
    \ind{\density{e}{G_i} < \tau_i}
    \\&\ge 
    \sum_{e \in B_i}
    \ind{\density{e}{G_i} < \tau_i}
    - 
    |B_i \cap D|
   \\&=
     \abs{B_i} - 
    \abs{\FilterF{}(B_i, G_i , \tau_i) }
    - |B_i \cap D|
    \\&\overset{(b)}{\ge} 
    \abs{B_i} - 
    \abs{\FilterF{}(B_i, G_i, \tau_i) }
    - \epsDel |B_i|
  \end{align*}
  where for $(a)$, we have used the fact that
  $L_i \backslash \Lc_i \subseteq D$, and for $(b)$, we have used Lemma \ref{lm:reconstruction_condition}.
  Therefore, 
  \begin{align*}
    \Ex{P_i - \bP_{i+1} \big| \sigmaPre_i}
     + \frac{\epsDel}{2} |B_i|
    &\ge
    \frac{1}{2}\Exu{m_i \sim \bm_i | \sigmaPre_i, \bm_i \in M_i^*}{\Ex{P_i - \bP_{i+1} \big| 
    \sigmaPre_i, \bm_i = m_i}}
    + \frac{\epsDel}{2} |B_i|
    \\&\ge
    \frac{1}{2}\Exu{m_i \sim \bm_i | \sigmaPre_i, \bm_i \in M_i^*}{
      \abs{B_i} - 
    \abs{\FilterF{}(B_i, G_i, \tau_i) 
    }
    \;
    \Bigg|
    \;
    \sigma^{-m}_i, \bm_i = m_i
    }
    \\&\ge
    \frac{1}{2} \cdot \frac{\epsSamp}{2} |B_i|
  \end{align*}
  where the first inequality follows from
  \eqref{lakjsdf},
  and
  the final inequality follows from Lemma \ref{lm:uniform_lazy}, together with Definition~\ref{def:suitable_sample_size}.

  Since
  $\epsDel \le \epsSamp/16$, it follows that
  \begin{align*}
   \Ex{P_i - \bP_{i+1} \big| \sigmaPre_i}
   \ge 
   \Omega(\epsSamp) \cdot |B_i|.
  \end{align*}

  We can conclude from Lemma~\ref{lm:reconstruction_condition} that
  $|\Lc_i| \le |\Lp_i| \le \frac{3}{2}\abs{L_i} \le 2\abs{L_i}$. Therefore,
  \begin{align*}
    \Ex{P_i - \bP_{i+1} | \sigmaPre_i}
    \ge \Omega(\epsSamp) \cdot |B_i|
    &\ge \Omega(\frac{\epsSamp}{\log_{1+\epsBuck}(\eta)\log_{1+\epsBuck}(\rho)}) \cdot |L_i|
    \\&\overset{(a)}{\ge} \Omega(\frac{\epsSamp}{\log_{1+\epsBuck}(\eta)\log_{1+\epsBuck}(\rho)}) \cdot |\Lc_i|
    \\&\overset{(b)}{\ge} \Omega(\frac{\epsSamp}{\log^2_{1+\epsBuck}(\eta)\log_{1+\epsBuck}(\rho)}) \cdot |P_i|
    \\&\overset{(c)}{\ge} \Omega(\frac{\epsilon^4}{\log^3(\eta)}) P_i
  \end{align*}
  where for $(a)$, we have used $|\Lc_i| \le 2\abs{L_i}$,
  for $(b)$ we have used Lemma~\ref{lm:potential_properties},
  and for $(c)$ we have used the inequality 
  $\log(1+x) \ge \frac{x}{4}$ for $x < 1$
  and the assumption $\epsBuck \le 1/10$.

  Therefore, for some 
  $\eps' = \Omega(\frac{\epsilon^4}{\log^3(\eta)})$,
   we have
  \begin{align*}
      \Ex{\bP_{i+1} | \sigmaPre_i 
      }
      \le 
      (1-\eps')\cdot P_{i},
  \end{align*}
  As before, this further implies
  \begin{align*}
      \Ex{\bP_{i+1} | 
      \bP_{i}=P_i,
      \dots, \bP_1 = P_1
      }
      \le 
      (1-\eps')\cdot P_{i}.
  \end{align*}
  Formally,
  if $P_{i} \ne 0$, 
  we have $\bT \ge i$ given Lemma~\ref{lm:potential_properties}. Therefore,
  by iterated expectation,
  \begin{align*}
      \Ex{\bP_{i+1} | 
      \bP_{i}=P_i,
      \dots, \bP_1 = P_1
      }
      &=
      \Ex{\bP_{i+1} | 
      \bP_{i}=P_i,
      \dots, \bP_1 = P_1, 
      \bT \ge i
      }
      \\&=
      \Ex{
      \Ex{\bP_{i+1} | 
      \bP_{i}=P_i,
      \dots, \bP_1 = P_1, 
      \bT \ge i,
      \bH_i = H_i
      }
      }
      \\&=
      \Ex{
      \Ex{\bP_{i+1} | 
      \bT \ge i,
      \bH_i = H_i
      }
      }
      \\&\le (1-\eps')P_i
  \end{align*}
  as claimed. If $P_i = 0$, then
  \eqref{qjhwq78} holds trivially because
  $P_{i+1}\le P_{i} = 0$. 
  Note however that
  \begin{align*}
      \frac{1}{\eps'} = \poly(\log(\eta), \frac{1}{\epsilon})
  \end{align*}
  and
  \begin{align*}
    P_1 \le \mO(\log_{1+\epsBuck}(\eta)) |\Rc_1|.
  \end{align*}
  The claim now follows from Lemma~\ref{lm:auxil}.
\end{proof}

\subsection{Proof of Theorem \ref{thm:bound_num_query_amor}}
  We start by bounding the count of "direct" queries, which come from insertions and deletions. Here, we don't count queries made indirectly through \ReconstructF{}. Each insertion or deletion can result in at most $\mO(\bT)$ queries, where $\bT$ is the number of levels during the update. According to Lemma~\ref{lm:num_level_total}, this is capped at $\poly(\log(n), \log(\eta), \frac{1}{\epsilon})$ because $|\Lc_1| \le n$.

Moving on to \enquote{indirect} queries made by \ReconstructF{}, we charge the cost of each $\ReconstructF{}(i)$ call to the updates causing it. If $\ReconstructF{}(i)$ is triggered by an insertion, its cost is charged to $\Lp_{i} \backslash L_i$, and if by a deletion, it's charged to $B_i \cap D$.
Each time $\ReconstructF{}(i)$ is called for some $i$, the expected number of queries is $|\Lc_i| \poly(\log(|\Lc_i|), \log(\eta), \frac{1}{\epsilon})$. However, this cost is spread across at least $\frac{|\Lc_i|}{\poly(\log(\rho), \log(\eta), \frac{1}{\eps})}$ updates due to the reconstruction conditions (The lower bound is chosen considering the reconstruction condition of deletion and size of $B_i$ and will clearly also hold if $\ReconstructF{}(i)$ is triggered by an insertion, because in that case $|\Lp_{i} \backslash L_i|$ is at least $\frac{1}{3}|\Lp_i| \geq \frac{1}{3}|\Lc_i|$). Hence, the cost of each charge is at most $\poly(\log(\eta),\frac{1}{\eps})$ (note that $|\Lc_i| \leq n$, and $\eta$ has both $\rho$ and $n$ as factors). Now, since each update can be charged by  $\ReconstructF{}(i)$ only once and only if when the update happens $T > i$ and level $i$ gets affected, we can say each update is charged at most $\poly(\log(\eta), \frac{1}{\eps})$ for each of the levels it affects. And as the expected number of levels during the update is at most $\poly(\log(\eta), \frac{1}{\eps})$ by Lemma~\ref{lm:num_level_total}, the claim follows. It's important to note that the random bits used to limit the expectation of $\bT$ and the ones used to limit the queries for each reconstruction are separate. Since the value $\bT$ is known at the update time, it relies on the random bits used before the update. In contrast, the number of queries for each $\ReconstructF{}$ depends on random bits used after (or at the time of) the update.

\section{Invariant proofs}
\label{app:invar_proofs}
The following two lemmas holds given the conditions in the algorithm for insertion and deletion.
\begin{lemma}\label{lm:invariant_filter}
    For all $i \in [T+ 1]$,
    $\hat{L}_i = \filter(\hat{L}_{i-1}, G_i, \tau)$.
\end{lemma}
\begin{proof}
    The lemma holds by construction of $L_i$ in 
    \ReconstructF{} and is preserved by insertion and deletion.
    given the conditions for adding an element to $L_i$. 
\end{proof}
\begin{lemma}\label{lm:invariant_subset}
    For all $i \in [T + 1]$,
    $\overline{L}_{i} \subseteq \overline{L}_{i - 1}$.
\end{lemma}
\begin{proof}
    The lemma holds by construction of $L_i$ in 
    \ReconstructF{}.
    Additionally, it is preserved by insertion and deletion because if $L_i$ is reconstructed, then $L_{i+1}$ is reconstructed as well.
\end{proof}
\begin{lemma}
\label{lm:reconstruction_condition}
    For all $i \in [T]$,
    $|B_i \cap D| \le \epsDel |B_i|$
    and $|\overline{L}_i| \le \frac{3}{2} |L_i|$
\end{lemma}
\begin{proof}
    The lemma holds by the reconstruction condition for insertion and deletion.
\end{proof}

\begin{lemma}\label{lm:final_level_stronger}
    $\hat{L}_{T+1} = \overline{L}_{T+1} = L_{T+1} =  \emptyset$ and
    $\hat{L}_i, L_i, \overline{L}_i \ne \emptyset$ for $i \in [1, T]$.
\end{lemma}
\begin{proof}
    The lemma holds after invoking \ReconstructF{} by definition of $T$. Additionally, it is preserved by insertion and deletion because 
    if $\hat{L}_i$ for $i \in [T]$ becomes empty or $\overline{L}_{T+1}$ becomes non-empty then $\ReconstructF{}$ is invoked, ensuring that the statement holds again. 
\end{proof}

\begin{lemma}\label{lm:uniform_lazy}
  For any $i\ge 1$, 
  and any $H_i$ such that
  $\Pr{\bT \ge i, \bH_i = H_i} > 0$,
    \begin{align}
        \Pr{\bold{S}_{i} = S \mid \bold{T} \ge i, \bold{H}_i = H_i}
        = 
        \frac{1}{|X_i|}\ind{S \in X_i}
    \end{align}
    where $|X_i|$ denotes all sequences of length $m_i$ in $L_i$.    
\end{lemma}
\begin{proof}

    We first prove that \eqref{eq:invariant_uniform} holds immediately after a call to \ReconstructF{}.
    
    \begin{claim}\label{lm:uniform_right_after}
      Assume we call $\ReconstructF{}(j)$ for some $j\le i$ in a data structure with values
      $T^{-}, L_0^-, \dots$, satisfying $T^- \ge i$,
      obtaining the new (random) values
      $\bT, \bL_0, \dots$.
      Then for all sets $S$,
      \begin{equation*}
        \Pr{\bS_i = S | \bT \ge i, \bH_i = H_i} =
        \frac{1}{|X_i|} \cdot \ind{S \in X_i}
        .
      \end{equation*}
    \end{claim}
    \begin{proof}
    We observe that
    before $S_i$ is sampled
    in the \ReconstructF{} algorithm,
    all of the values compromising $H_i$ are already determined and will not 
    change after $S_i$ is sampled.
    Therefore, since the claim holds
    when $S_i$ is sampled (by construction of $S_i$), it holds afterwards as well.

    \end{proof}
    
    Next, we prove Lemma~\ref{lm:uniform_lazy}  by showing that \eqref{eq:invariant_uniform} is
    preserved after insertions and deletions.
  At the beginning of the stream, the lemma holds trivially since
  the event $\bT \ge i$ is impossible because of $T = 0 < 1 \le i$.
  Assuming the lemma holds before some update, we will show that it holds after the update as well.

  Formally, let's consider an update in the form of either inserting or deleting an element $v$. We denote the values of the data structure before the operation as $\bT^{-}, \bL_0^{-}, \dots$ and the values afterward as $\bT, \bL_0, \dots$. For a fixed $i$, we will prove that \eqref{eq:invariant_uniform} holds.

  The proof revolves around analyzing two cases based on whether $\LevelF{}(i)$ was triggered by the insertion or deletion operation. In the first case, result follows from Claim~\ref{lm:uniform_right_after}. In the second case, we rely on the induction hypothesis which asserted that \eqref{eq:invariant_uniform} held for the values $\bT^{-}, \bL_0^{-}, \dots$.

  We will now proceed with a formal proof.

Let $\breset{}_i$ be a random variable that takes the value $1$ if $\LevelF{}(i)$ was called because of the update and $0$ otherwise. Define the event $\sigma_i$ as $\bT \ge i \land \bH_i = H_i$. We need to prove that $\Pr{\bS_i = S | \sigma_i} = \frac{1}{|X_i|} \cdot \ind{S \in X_i}$.
By conditioning on $\breset{}_i$, we can express $\Pr{\bS_i=S | \sigma_i}$ as follows:
\begin{equation}
  \Pr{\bS_i = S| \sigma_i} = 
  \Exu{\reset{}_i \sim \breset{}_i | \sigma_i}{
    \Pr{\bS_i = S| \sigma_i, \breset{}_i=\reset{}_i}
  }.
  \label{eq:s_sigma_decompose_l}
\end{equation}

Thus, it suffices to prove

\begin{equation}
  \Pr{\bS_i = S| \sigma_i, \breset{}_i=\reset{}_i}
  =
  \frac{1}{|X_i|} \cdot \ind{S \in X_i}
  ,
  \label{eq:jul11_1815}
\end{equation}

for both $\reset{}_i=0$ and $\reset{}_i=1$. For $\reset{}_i=1$, the claim holds by Claim~\ref{lm:uniform_right_after} since, by definition, $\reset{}_i=1$ implies that \ReconstructF{}(j) was called for some $j \le i$.

  Therefore, we focus on the case of $\reset{}_i=0$. We begin by observing that $\sigma_i, \breset{}_i=0$ implies $\bT^- \ge i$. This is because $\breset{}_i=0$ indicates that $\ReconstructF{}(j)$ for any $j \le i$ was not called. Consequently, if $\bT^-$ were strictly less than $i$, then $\bT$ would equal $\bT^-$ (as $\ReconstructF{}$ can only be called for values up to $T+1$ and if $\ReconstructF{}$ is not called, then $T$ does not change). However, this is not possible since $\bT \ge i$.

Given that $\bT^- \ge i$, we can condition on the value of the \emph{previous history} of level $i$. More formally, we define the random variable $\bH_i^-$ as:

\begin{equation*}
  \bH_i^- := (\bLp_0^-, \bLp_1^-, \dots, \bLp_i^-, \bL_0^-, \bL_1^-, \dots, \bL_i^-, \bS_1^-, \dots, \bS_{i-1}^-, \bm_i^-).
\end{equation*}

By the law of iterated expectation, we express:

\begin{align}
  \notag
  \Pr{\bS_i = S| \sigma_i, \breset{}_i=0}
  &=
  \Pr{\bS_i = S| \sigma_i, \breset{}_i=0, \bT^-\le i}
  \\&=
  \Exu{H_i^- \sim \bH_i^- | \sigma_i, \breset{}_i=0
  , \bT^-\le i}{
    \Pr{\bS_i = S| \sigma_i, \breset{}_i=0, \bT^-\le i, \bH_i^- = H_i^-}
  }
  \label{eq:Jul10_2249}
\end{align}

where the expectation is taken over all $H_i^-$ with positive probability.

  We now observe that conditioned on $\bT^-\le i, \bH_i^-=H_i^-$, the value of $\breset{}_i$ always equals $0$. This is because $\reset{}_i$ is a function of $(\Lp_{0}, \dots, \Lp_{i}, L_0,\dots, L_i, D)$, which is determined by $H_i$. Note that $D$ is deterministic since it contains all the deleted elements in the stream and is independent of our algorithm. Therefore, since we only consider $H_i^-$ with positive probability, we can drop the conditioning on $\breset{}_i=0$ in the $\Pr{\bS_i = S| \sigma_i, \breset{}_i=0, \bT^-\le i, \bH_i^- = H_i^-}$ term of \eqref{eq:Jul10_2249} since it is redundant.

We can similarly drop $\sigma_i$. This is because, as $\breset{}_i=0$, the value of $\bH_i$ is deterministic conditioned on $\bH_i^-=H_i^-$. Notably, the values of $\bL_1, \dots, \bL_{i}$ are going to be $L_1^-, \dots, L_i^-$. The same can be said for $\bS_1, \dots, \bS_{i-1}$ and $\bm_i$. As for $\bLp_1, \dots, \bLp_{i}$, even though their value may be different from $\Lp_1^-, \dots, \Lp_{i}^-$, it is \emph{still deterministic} conditioned on $\bH_i^-=H_i^-$ as the decision to add elements in Line \ref{line:insert_add} is based on the values in $H_i^-$ only.

Given the above observation, we can rewrite $\Pr{\bS_i = S| \sigma_i, \breset{}_i=0}$ as:

\begin{align*}
  \Pr{\bS_i = S| \sigma_i, \breset{}_i=0}
  &=
  \Exu{H_i^- \sim \bH_i^- | \sigma_i, \breset{}_i=0}{
    \Pr{\bS_i = S| \bT^-\le i, \bH_i^- = H_i^-}
  }
\end{align*}

We can further replace $\bS_i=S$ with $\bS_i^-=S$ as $\breset{}_i=0$ implies $\bS_i = \bS_i^-$.
  
  It follows that
  \begin{align*}
    \Pr{\bS_i = S| \sigma_i, \breset{}_i=0}
    &=
    \Exu{H_i^- \sim \bH_i^- | \sigma_i, \breset{}_i=0}{
      \Pr{\bS_i^- = S| \bT^-\le i, \bH_i^- = H_i^-}
    }
    \\&\overset{(a)}{=}
    \Exu{H_i^- \sim \bH_i^- | \sigma_i, \breset{}_i=0}{
      \frac{1}{|X_i^{-}|} \cdot \ind{S \in X_i^{-}}
    }
    \\&\overset{(b)}{=}
    \Exu{H_i^- \sim \bH_i^- | \sigma_i, \breset{}_i=0}{
      \frac{1}{|X_i|} \cdot \ind{S \in X_i}
    }
    \\&=
    \frac{1}{|X_i|} \cdot \ind{S \in X_i}
  \end{align*}
  where for $(a)$, we have used the induction assumption,
  and for $(b)$ we have used the fact that
  $X_i^{-}=X_i$ because of $\breset{}_i=0$. 
  We have therefore proved \eqref{eq:jul11_1815} 
  for both $\reset{}_i=0$ and $\reset{}_i=1$, which completes the proof given \eqref{eq:s_sigma_decompose_l}.
\end{proof}

\begin{lemma}\label{lm:bound_countF}
  Consider a call to $\CalcSampleCountF{}(L', G', \tau')$ with values satisfying $\FilterF{}(L', G', \tau') = L'$ and $L' \ne \emptyset$. The number of queries made by $\CalcSampleCountF{}$ is bounded by
  \begin{math}
    \mO\left(|L'| \cdot \frac{\log(n)}{\epsSamp^3}\right).
  \end{math}
  Furthermore, the output is a suitable sample size (Definition \ref{def:suitable_sample_size}) with probability at least $1 - \mO(\frac{\epsSamp}{n^{10}})$.
\end{lemma}

\begin{proof}
  To bound the number of queries, note that there will be $\constt$ calls to \apprev{}, and each call makes $|L'|$ queries, implying the first part of the lemma's statement.

  Focus on proving that the output is suitable with a probability of at least $1 - \mO(\frac{\epsSamp}{n^{10}})$.

  For a number $r$, define the value $\bX(r)$ as explained in Definition~\ref{def:suitable_sample_size}.
  For the number $m'$ defined in \CalcSampleCountF{}, we want to prove that $\Ex{\bX(r)} \ge 1-2\epsilon \text{ for all }r \in [1, m'-1]$ and $\Ex{\bX(m')} \le 1-\frac{\epsilon}{2}$ with probability at least $1 - \mO(\frac{\epsSamp}{n^{10}})$.

  Hoeffding's inequality implies that for any $r \in [1, m']$
      \begin{align*}
          \Pr{\abs{\frac{\sum_{i=1}^t X_i(r)}{t} - \Ex{\bX(r)} } \ge \frac{\epsSamp}{2}} \le 2e^{-\frac{t\epsSamp^2}{4}} \le
          \frac{\epsSamp}{n^{12}}\text{,}
      \end{align*}
      where the second inequality follows from the assumption
      $t = \ceil{\constt}$. 
      Applying union bound over $r$ implies that
      \begin{align}
          \Pr{\forall r \in [1, m']: \abs{\frac{\sum_{i=1}^t X_i(r)}{t} - \Ex{\bX(r)} } \ge \frac{\epsSamp}{2}} \le \frac{\epsSamp}{n^{{10}}}\text{.}
          \label{eq:oct2_1229}
      \end{align}
      According to Line~\ref{line:find_m'},
      $\frac{\sum_{i=1}^t X_i(r)}{t} \ge 1-\epsSamp$ for all $r < m'$ and $\frac{\sum_{i=1}^t X_i(r)}{t} \le 1 - \epsSamp$ for $r=m'$. 
      Together with Equation \eqref{eq:oct2_1229}, this implies that
      with probability at least $1-\frac{\epsSamp}{n^{10}}$,
      \begin{align*}
          \Ex{\bX(r)} \ge 1- 2\epsSamp\text{ for all }j \in [1, m' - 1] \text{ and } \Ex{\bX(m')} \le 1-\frac{\epsSamp}{2}\text{.}
      \end{align*}
      This proves that $m'-1$ returned by \CalcSampleCountF{} is a suitable sample size with probability at least $1 - \mO(\frac{\epsSamp}{n^{10}})$.

\end{proof}

\begin{lemma}\label{lm:bound_countF_always}
  For any $i \ge 1$ and at any point in the stream, the following inequality holds:
  \begin{equation*}
    \Pr{\bm_i \in M_i^* | \bT \ge i, \bhistPre_{i} = \histPre_{i}} \ge 1 - \mO(\frac{\epsSamp}{n^{10}})
  \end{equation*}
\end{lemma}
\begin{proof}
  We establish the claim through induction on the update stream. Initially, the claim trivially holds as $\bT \ge i$ is impossible.

  Assuming the lemma's statement holds before an update, we demonstrate that it holds for the new values as well. The superscript $^{-}$ denotes values before the insertion, e.g., $T^{-}$ signifies the number of levels before insertion, and $T$ denotes the number of levels after insertion.

  Let $\breset_i$ be a random variable set to 1 if $\ReconstructF{}(j)$ is called for some $j \le i$ and 0 otherwise. We show that
  \begin{equation*}
    \Pr{\bm_i \in M_i^* | \bT \ge i, \bhistPre_{i} = \histPre_{i}, \breset_i=\reset_i} \ge 1 - \mO(\frac{\epsSamp}{n^{10}})
  \end{equation*}
  holds for both $\reset_i=0$ and $\reset_i=1$. For $\reset_i=1$, the claim follows from Lemma \ref{lm:bound_countF}. As for $\reset_i=0$, it implies $\bT^- \ge i$ because if $\bT^- < i$, then $\bT$ would have been equal to $\bT^-$ due to $\breset_i=0$, signifying that $\ReconstructF{}(j)$ was not invoked for any $j \le i$.

  Let $(\histPre_{i})^{-}$ denote the value of pre-count history before the update. By the law of iterated expectation, it suffices to show
  \begin{equation}
    \Pr{\bm_i \in M_i^* | \bT \ge i, \bT^- \ge i, \bhistPre_{i} = \histPre_{i}, (\bhistPre_{i})^{-} = (\histPre_{i})^{-}, \breset_i=0} \ge 1 - \mO(\frac{\epsSamp}{n^{10}})
    \label{eq:jul12_1507}
  \end{equation}
  for all $(\histPre_{i})^{-}$ with positive probability conditioned on $\bT\ge i,\bhistPre_{i}=\histPre_{i}, \breset_i=0$. We can drop the $\bT\ge i,\bhistPre_{i}=\histPre_{i}, \breset_i=0$ from the condition since they are implied by $(\bhistPre_{i})^{-}=(\histPre_{i})^{-}$. Furthermore, we can replace $\bm_i \in M_i^*$ with $\bm_i^{-} \in (M_i^*)^{-}$, where $(M_i^*)^{-}$ is the set of suitable sample counts for level $i$ before the update, as determined by $(\histPre_{i})^{-}$. This is because both $\bm_i$ and $M_i^*$ remain unchanged through the update (note that $M_i^*$ is a function of $L_i, G_{i-1}$, and $(L_i, G_{i-1})=(L_i^{-}, G_{i-1}^{-}$)). This transformation reduces \eqref{eq:jul12_1507} to the induction hypothesis, concluding the proof.
\end{proof}

\end{document}

%% file: macros.tex
\newcommand{\mO}{O}

\newcommand{\EXP}{\mathbb{E}}
\newcommand{\PROB}{\textnormal{Pr}}
\newcommand{\IND}{\mathds{1}}
\renewcommand{\Pr}[1]{{\PROB \sbr{#1}}}

\newcommand{\Ex}[1]{{\EXP \sbr{#1}}}
\newcommand{\Exu}[2]{\ensuremath{\EXP_{#1}\left[#2\right]}}
\newcommand{\ind}[1]{{\IND \left\{ #1 \right\}}}
\newcommand{\eps}{\varepsilon}

\newcommand{\floor}[1]{\left\lfloor\,#1\,\right\rfloor}
\newcommand{\ceil}[1]{\left\lceil\,#1\,\right\rceil}
\newcommand{\rbr}[1]{\left(\,#1\,\right)}
\newcommand{\sbr}[1]{\left[\,#1\,\right]}
\newcommand{\cbr}[1]{\left\{\,#1\,\right\}}
\DeclarePairedDelimiter\abs{\lvert}{\rvert}%

\newcommand{\cost}{\textsc{Cost}}
\newcommand{\weight}{\ensuremath{w}}

\newcommand{\ground}{\ensuremath{\mathcal{V}}}
\newcommand{\solution}{\ensuremath{S}}
\newcommand{\opt}{\ensuremath{opt}}

\newcommand{\optcost}{\ensuremath{\textsc{OPT}_{\text{cost}}}}

\newcommand{\Lp}{\ensuremath{\overline{L}}}

\newcommand{\bS}{\ensuremath{\mathbf{S}}}

\newcommand{\bm}{\ensuremath{\mathbf{m}}}
\newcommand{\bl}{\ensuremath{\mathbf{l}}}
\newcommand{\bg}{\ensuremath{\mathbf{g}}}
\newcommand{\bM}{\ensuremath{\mathbf{M}}}
\newcommand{\bP}{\ensuremath{\mathbf{P}}}

\newcommand{\bH}{\ensuremath{\mathbf{H}}}
\newcommand{\bX}{\ensuremath{\mathbf{X}}}

\newcommand{\bLp}{\ensuremath{\mathbf{\Lp}}}

\newcommand{\bT}{\ensuremath{\mathbf{T}}}

\newcommand{\bL}{\ensuremath{\mathbf{L}}}
\newcommand{\Rc}{\ensuremath{\widehat{R}}}

\newcommand{\Lc}{\ensuremath{\widehat{L}}}

\newcommand{\algorithmicbreak}{\textbf{break}}
\newcommand{\Break}{\algorithmicbreak}
\newcommand{\algorithmiccontinue}{\textbf{continue}}
\newcommand{\Continue}{\algorithmiccontinue}

\newcommand{\ReconstructF}{{\textsc{Reconstruct}}}
\newcommand{\FilterF}{{\textsc{Filter}}}
\newcommand{\filter}{\FilterF}
\newcommand{\epsBuck}{{\epsilon}}
\newcommand{\CalcSampleCountF}{{\textsc{CalcSampleSize}}}
\newcommand{\InitF}{{\textsc{Init}}}
\newcommand{\density}[2]{d(#1|#2)}

\newcommand{\epsDel}{{\epsilon_{\textnormal{del}}}}
\newcommand{\epsSamp}{\epsilon}
\newcommand{\poly}{\textnormal{poly}}

\newcommand{\InsertF}{{\textsc{Insert}}}
\newcommand{\DeleteF}{{\textsc{Delete}}}

\newcommand{\LevelF}{\ReconstructF}

\newcommand{\reset}{\textnormal{Reset}}
\newcommand{\breset}{\textbf{Reset}}
\newcommand{\histPre}{H^{\text{pre}}}
\newcommand{\sigmaPre}{\sigma^{\text{pre}}}
\newcommand{\bhistPre}{\mathbf{H}^{\text{pre}}}
\newcommand{\apprev}{{\textsc{ApplyAndRevert}}}
\newcommand{\constt}{{4\frac{1}{\epsilon^2}\log(\frac{n^{12}}{\epsilon})}}

\newcommand{\buck}{\text{buck}}
\newcommand{\update}{{\textsc{update}}}